\documentclass[reqno]{amsart}

\usepackage{amsaddr}
\usepackage{amsmath}
\usepackage{amssymb}
\usepackage{amsthm}
\usepackage{bm}
\usepackage{eucal}
\usepackage{fullpage}
\usepackage{graphicx}
\usepackage{lineno}
\usepackage{mathabx}
\usepackage{mathrsfs}
\usepackage[numbers,sort&compress,square]{natbib}
\usepackage{setspace}
\usepackage{subfig}
\usepackage[all,cmtip]{xy}
\usepackage{color}

\theoremstyle{definition}
\newtheorem{lemma}{Lemma}
\newtheorem{remark}{Remark}
\newtheorem{proposition}{Proposition}
\newtheorem{theorem}{Theorem}

\newcommand{\eq}[1]{\textbf{Eq.~\ref{eq:#1}}}
\newcommand{\fig}[1]{\textbf{Fig.~\ref{fig:#1}}}
\newcommand{\thm}[1]{Theorem~\ref{thm:#1}}

\title{Public goods games in populations with fluctuating size}

\author{Alex McAvoy$^{1}$, Nicolas Fraiman$^{2}$, Christoph Hauert$^{3}$, John Wakeley$^{4}$, and Martin A. Nowak$^{1,4,5}$}

\address{\small$^{1}$Program for Evolutionary Dynamics, Harvard University, Cambridge, MA 02138 \\ $^{2}$Department of Statistics and Operations Research, University of North Carolina at Chapel Hill, Chapel Hill, NC 27599 \\$^{3}$Department of Mathematics, University of British Columbia, 1984 Mathematics Road, Vancouver, BC, Canada V6T 1Z2 \\ $^{4}$Department of Organismic and Evolutionary Biology, Harvard University, Cambridge, MA 02138 \\ $^{5}$Department of Mathematics, Harvard University, Cambridge, MA 02138}

\begin{document}

\allowdisplaybreaks

\begin{abstract}
Many mathematical frameworks of evolutionary game dynamics assume that the total population size is constant and that selection affects only the relative frequency of strategies. Here, we consider evolutionary game dynamics in an extended Wright-Fisher process with variable population size. In such a scenario, it is possible that the entire population becomes extinct. Survival of the population may depend on which strategy prevails in the game dynamics. Studying cooperative dilemmas, it is a natural feature of such a model that cooperators enable survival, while defectors drive extinction. Although defectors are favored for any mixed population, random drift could lead to their elimination and the resulting pure-cooperator population could survive. On the other hand, if the defectors remain, then the population will quickly go extinct because the frequency of cooperators steadily declines and defectors alone cannot survive. In a mutation-selection model, we find that (i) a steady supply of cooperators can enable long-term population survival, provided selection is sufficiently strong, and (ii) selection can increase the abundance of cooperators but reduce their relative frequency. Thus, evolutionary game dynamics in populations with variable size generate a multifaceted notion of what constitutes a trait's long-term success.
\end{abstract}

\maketitle

\section{Introduction}
The emergence of cooperation is a prominent research topic in evolutionary theory. The problem is usually formulated in such a way that it pays to exploit cooperators, yet the payoff to one cooperator against another is greater than the payoff to one defector against another \citep{axelrod:Science:1981}. In spite of this conflict of interest, cooperation is broadly observed in nature, and various mechanisms have been put forth to explain its evolution \citep{nowak:Science:2006}. In fact, the question of how cooperators may proliferate in social situations is one of the main concerns of evolutionary game theory, a framework that models cooperation and defection as strategies of a game.

Evolutionary game-theoretic models typically involve a number of assumptions. In this study, we are concerned with two potentially restrictive ones: (i) the population size is infinite or (ii) the population size is finite but fixed and unaffected by evolution. While the classical replicator equation \citep{taylor:MB:1978,hofbauer:JTB:1979,hofbauer:CUP:1998} can be used to model large populations that fluctuate in size \citep{hauert:PRSB:2006}, replicator dynamics capture only the relative frequencies of the strategies. Even stochastic models that account for populations of any finite size, such as the Moran model or the Wright-Fisher model and their game-theoretic extensions, usually require the number of players to remain fixed over time \citep{moran:MPCPS:1958,nowak:Nature:2004,taylor:BMB:2004,lieberman:Nature:2005,ohtsuki:Nature:2006,taylor:Nature:2007,szabo:PR:2007,tarnita:PNAS:2009,nowak:PTRSB:2009,hauert:JTB:2012,debarre:NC:2014}. Here, we explore the evolutionary dynamics of cooperation in social dilemmas when the population can fluctuate in size and even go extinct.

Branching processes have a rich history in theoretical biology \citep[see][]{kimmel:S:2015} and are a natural way to model populations that vary in size. A number of recent works have considered non-constant population size within evolutionary game theory. \citet{hauert:PRSB:2006} treat ecological dynamics in evolutionary games by modifying the replicator equation to account for population density and show that fluctuating density can lead to coexistence between cooperators and defectors. \citet{melbinger:PRL:2010} illustrate how the decoupling of stochastic birth and death events can lead to transient increases in cooperation. By allowing a game to influence carrying capacities, \citet{novak:JTB:2013} demonstrate that variable density regulations can change the stability of equilibria relative to the replicator equation. Furthermore, demographic fluctuations can act as a mechanism to promote cooperation in public goods games \citep{constable:PNAS:2016} and indefinite coexistence (as opposed to fixation) in coexistence games \citep{ashcroft:JTB:2017}. Fluctuating size in a Lotka-Volterra model also leads to different growth rates for isolated populations of cooperators and defectors \citep{huang:PNAS:2015}, and even when the two competing types are neutral at the equilibrium size, fluctuations can still give one type a selective advantage over the other \citep{chotibut:JSP:2017}. When traits have the same growth rate, these fluctuations also affect a mutant's fixation probability \citep{czuppon:preprint:2017}.

Here, we develop a branching-process model based on the Wright-Fisher model \citep{fisher:OUP:1930,wright:G:1931} for a population with non-overlapping generations in which trait values of offspring are sampled from the previous generation depending on the success of individuals (parents) in a sequence of interactions \citep{ewens:S:2004,imhof:JMB:2006}. Success is quantified in terms of payoffs, which come from a game and represent competition between the different types, or strategies. Usually, the Wright-Fisher process is defined such that every subsequent generation has exactly the same size as the first generation. We consider a variant of this model for populations that fluctuate in size, in which each individual has a Poisson-distributed number of surviving offspring, with an expected value determined by payoffs from interactions in a game.

Recently, \citet{houchmandzadeh:B:2015} considered a model similar to the one we study here, but under the assumption that the population size in the next generation, $N\left(x\right)$, is a deterministic function of the fraction of cooperators in the present generation, $x$. The update rule then has essentially two stages: (i) determine the population size of the next generation, $N\left(x\right)$, and (ii) sample $N\left(x\right)$ offspring from the previous generation using the standard Wright-Fisher rule \cite{houchmandzadeh:B:2015}. In contrast, the model we treat has a stochastic population size that does not need to be prespecified. Moreover, it depends on the numbers of both cooperators and defectors in the current generation, not just on the fraction of cooperators. As mentioned above, we also allow for the possibility that the entire population goes extinct.

We use the public goods game to study the evolution of cooperation in an unstructured population. Cooperators maintain a shared resource or public good, with a cost, $w$, to their fecundity. Defectors neither help maintain the public good nor incur a cost. The resource is distributed evenly among all individuals in the population, but its per-capita effect on fecundity can be greater than the per-capita cost of its production \cite{sigmund:PUP:2010}. A multiplication factor, $R>1$, quantifies this return on the investment made by cooperators toward production of the good. In this model, everyone is better off when the whole population consists of cooperators, but defectors can benefit from cooperation without paying the cost. 

We show that when the population size can fluctuate, selection can be essential for the survival of the population as a whole. In our model, population growth and decline are influenced by the public goods game but also by a baseline reproductive capacity, $f_{N}$, which is the same for all individuals and which primarily acts to constrain runaway growth. Even when cooperators are less frequent than defectors in the mutation-selection equilibrium, there can be an optimal cost of cooperation, $w^{\ast}$, depending on $f_{N}$, at which (i) the population does not immediately go extinct, with the numbers of cooperators and defectors each fluctuating around equilibrium values, and (ii) the frequency of cooperators is maximized subject to (i). In other words, cooperation can be favored by selection at a positive cost of cooperation when there is demographic stochasticity, which marks a departure from the behavior of models with fixed size.

Furthermore, even when the population would survive due to the baseline reproductive capacity alone, selection can increase the number of cooperators while at the same time decreasing their frequency. In models where the population size is assumed to be fixed, cooperators are less frequent than defectors if and only if cooperators are less abundant than defectors. However, this equivalence breaks down when the population size can fluctuate because the frequency of a strategy is determined by both its abundance and the population size. Thus, the evolutionary success of a strategic type depends on more than just the strategy.

\section{Description of the model}
We use the term ``reproductive capacity'' rather than ``fitness'' \citep[see][]{doebeli:eLife:2017} to refer to the expected number of offspring of an individual. In a growing population, the average reproductive capacity is greater than one. In a shrinking population, it is less than one. In a population of fixed size or a population at its carrying capacity, the average reproductive capacity is equal to one. If different individuals in the same population have different reproductive capacities, some individuals have a selective advantage over others.

\subsection{Update rule}
We assume that individuals reproduce asexually, so our model corresponds to a model of haploid genetic transmission. In the standard Wright-Fisher process, the population has fixed size, $N$. Thus, in a game with two strategies, $C$ (``cooperate'') and $D$ (``defect''), the state of the population is determined by number of cooperators, $x_{C}$, or by their relative frequency, $x_{C}/N$. If $F_{C}=F_{C}\left(x_{C}\right)$ and $F_{D}=F_{D}\left(x_{C}\right)$ give the reproductive capacities of cooperators and defectors, respectively, in the state with $x_{C}$ cooperators, then the probability of transitioning to the state with $y_{C}$ cooperators (provided $0\leqslant y_{C}\leqslant N$) is
\begin{linenomath}
\begin{align}\label{eq:wfUpdateRule}
\mathbf{P}\left(y_{C} \mid x_{C}\right) &= \binom{N}{y_{C}}\left(\frac{x_{C}F_{C}}{x_{C}F_{C}+\left(N-x_{C}\right) F_{D}}\right)^{y_{C}}\left(\frac{\left(N-x_{C}\right) F_{D}}{x_{C}F_{C}+\left(N-x_{C}\right) F_{D}}\right)^{N-y_{C}} .
\end{align}
\end{linenomath}
In other words, the cooperators in one generation are sampled from the previous generation according to a binomial distribution with mean $Nx_{C}F_{C}/\left(x_{C}F_{C}+\left(N-x_{C}\right) F_{D}\right)$. One biological interpretation for this transition rule is the following: Each player in one generation produces a large number of gametes from which the surviving offspring in the next generation are selected. These offspring are sampled at random, weighted by the success of the parents in competitive interactions, subject to a constant population size.

In treating populations that fluctuate in size, we drop the assumed dependence that $y_{D} = N - y_{C}$ which is implied above, but continue to hold that generations are non-overlapping. Let $F_{C}=F_{C}\left(x_{C},x_{D}\right)$ and $F_{D}=F_{D}\left(x_{C},x_{D}\right)$ give the reproductive capacities of cooperators and defectors, respectively, when the current generation is in state $\left(x_{C},x_{D}\right)$. We assume that the number of offspring per individual follows a Poisson distribution, with parameter $F_{C}$ for cooperators and parameter $F_{D}$ for defectors. Then the probability of transitioning from state $\left(x_{C},x_{D}\right)$ to state $\left(y_{C},y_{D}\right)$ in one generation is
\begin{linenomath}
\begin{align}\label{eq:wfbUpdateRule}
\mathbf{P}\left( y_{C},y_{D} \mid x_{C},x_{D} \right) &= \left(\frac{\left(x_{C}F_{C}\right)^{y_{C}}e^{-x_{C}F_{C}}}{y_{C}!}\right)\left(\frac{\left(x_{D}F_{D}\right)^{y_{D}}e^{-x_{D}F_{D}}}{y_{D}!}\right) .
\end{align}
\end{linenomath}
\eq{wfbUpdateRule} reduces to \eq{wfUpdateRule} when the population size is fixed and equal to $N$ (see \citep{haccou:CUP:2005} and also \ref{sec:appendixA}).

The transition probabilities of \textbf{Eqs. \ref{eq:wfUpdateRule}--\ref{eq:wfbUpdateRule}} do not take into account errors in strategy transmission, i.e. mutations. In what follows, we assume that when an individual reproduces, the offspring acquires a random strategy with probability $u\geqslant 0$. Thus, with probability $1-u$, the offspring acquires the strategy of the parent and with probability $u$, becomes either a cooperator or a defector (uniformly at random). For simplicity (and by convention \citep[e.g. see][]{tarnita:JTB:2009}), we assume symmetric mutation, with $C\rightarrow D$ as likely as $D\rightarrow C$.

Using the binomial distribution with parameter $q$ and density function $b_{q}\left(n,k\right) :=\binom{n}{k}q^{k}\left(1-q\right)^{n-k}$, the mutation rate, $u$, is incorporated into the transition rule defined by \eq{wfbUpdateRule} as follows:
\begin{linenomath}
\begin{align}\label{eq:wfbUpdateRuleMutation}
\mathbf{P}^{u}\left( y_{C},y_{D} \mid x_{C},x_{D} \right) &= \sum_{z_{C}=0}^{y_{C}+y_{D}} \mathbf{P}\left( z_{C},y_{C}+y_{D}-z_{C} \mid x_{C},x_{D} \right) \nonumber \\
&\quad\quad \times\sum_{k=\max\left\{0,z_{C}-y_{C}\right\}}^{\min\left\{z_{C},y_{D}\right\}}b_{u/2}\left(z_{C},k\right) b_{u/2}\left(y_{C}+y_{D}-z_{C},k+y_{C}-z_{C}\right) .
\end{align}
\end{linenomath}
In words, we sum over all transitions defined by \eq{wfbUpdateRule} such that, after mutations are accounted for, there are $y_{C}$ cooperators and $y_{D}$ defectors. Note that mutations themselves do not affect the population size.

We refer to the process with transitions governed by \textbf{Eqs. \ref{eq:wfbUpdateRule}--\ref{eq:wfbUpdateRuleMutation}} as a ``Wright-Fisher branching process.'' Branching processes of this sort have been treated elsewhere \citep[see][]{haccou:CUP:2005}, notably with Poisson-distributed offspring counts but frequency-independent reproductive capacities \citep{haccou:TPB:1996}. Branching processes have also been considered in models with both density-dependent \citep{lambert:AAP:2005} and frequency-dependent \cite{wild:BMB:2010,bao:TPB:2012} reproductive capacities. We consider a Wright-Fisher branching process in which the reproductive capacities of cooperators and defectors in \textbf{Eqs. \ref{eq:wfbUpdateRule}--\ref{eq:wfbUpdateRuleMutation}} are equal to a baseline reproductive capacity times a factor which depends on the outcome of a public goods game.

\subsection{Baseline reproductive capacity}
The standard Wright-Fisher model, and variants like that in \citep{houchmandzadeh:B:2015}, assume that population regulation is very strong or deterministic. This may be a good approximation for large populations and those close to carrying capacity \citep[but see][]{chotibut:JSP:2017}. However, fully stochastic treatments are warranted for populations that fluctuate in size and may go extinct. In the Wright-Fisher branching process we consider, population regulation is achieved through a balance of players' payoffs in an evolutionary game and a baseline reproductive capacity which represents all other ecological factors. The dynamics depend on the numbers of cooperators and defectors, not just on their relative frequencies. The baseline reproductive capacity is a function of the total population size and is the same for every individual. It captures the ecological constraints which keep populations from growing without limit. 

Let $f_{N}$ be the per-capita reproductive capacity (again, the expected number of offspring) applicable to all individuals when the population size is $N$. We assume that $f_{N}$ is a non-increasing function of the population size, $N$, so that larger populations lead (in general) to lower per-capita reproductive capacities due to ecological constraints. This baseline reproductive capacity is the fluctuating-size analogue of the ``background fitness'' that is typically used in models with fixed population size \citep{nowak:PTRSB:2009}.

While the baseline reproductive capacity does not vary from player to player, it can depend on the number of players in the population, $N$. If there are limited resources and reproduction slows as the population grows, then $f_{N}$ is a decreasing function of $N$. An example we consider is $f_{N} = c_{K}+r\left(1-\frac{N}{K}\right)$ for some $c_{K}$, $r$, and $K$. In this case, $r$ reflects the growth rate when the population is small, and $c_{K}$ gives the reproductive capacity when the population has size $N=K$. To ensure $f_{N}$ is non-negative, we set $f_{N}=0$ whenever $c_{K}+r\left(1-\frac{N}{K}\right)\leqslant 0$. A second example we consider is one in which the baseline reproductive capacity is constant up to a threshold value of $N$ and decreasing thereafter, specifically with $f_{N} =1+r$ if $N\leqslant K$ and $f_{N} =\left(1+r\right)K/N$ if $N>K$ for some $r$ and $K$. For both of these functions, $f_{N}$ decays to $0$ as $N\rightarrow\infty$ (see \fig{baselineRates}).

\begin{figure}
\centering
\includegraphics[width=0.8\textwidth]{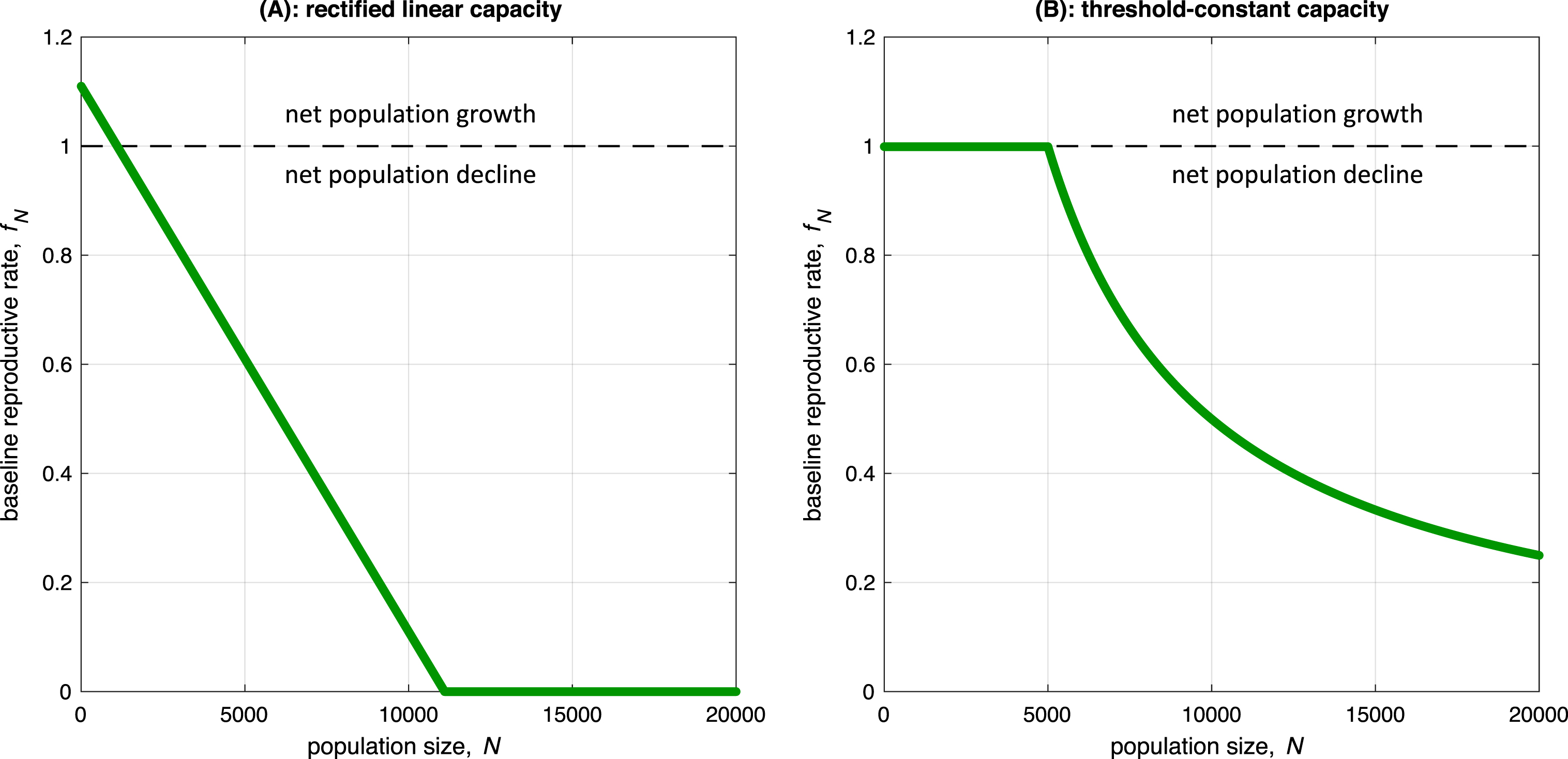}
\caption{Two examples of baseline reproductive capacities, $f_{N}$. In (A), $f_{N}=\max\left\{0,c_{K}+r\left(1-\frac{N}{K}\right)\right\}$ for $c_{K}=1.1$, $r=0.01$, and $K=100$. This function decreases linearly to $0$ and remains at $0$ for all sufficiently large $N$ since the reproductive capacity is (by definition) non-negative. This function is called ``rectified'' linear because of the constraint that $f_{N}\geqslant 0$. In (B), $f_{N}=\left(1+r\right)\min\left\{1,\frac{K}{N}\right\}$ for $r=-0.001$ and $K=5000$. The reproductive capacity is held constant and equal to $1+r$ whenever $1\leqslant N\leqslant K$. When the population size is greater than $K$, the resources contributing to reproductive capacities are no longer abundant and must be divided up among the individuals in the population, which gives the expression $f_{N}=\left(1+r\right)\frac{K}{N}$ for $N\geqslant K$. This function is called ``threshold-constant'' because it is constant up to the threshold population size $N=K$. Note that in (A), when the population size is initially small, there is net population growth since $f_{1}>1$. In (B), $f_{N}<1$ for every $N$, so any population evolving according to this function alone will (on average) shrink in the direction of extinction. This immediate extinction could be prevented by choosing $r>0$ instead of $r<0$ as in (B).\label{fig:baselineRates}}
\end{figure}

\subsection{Public goods game}
Consider a game with two strategies, $C$ (``cooperate'') and $D$ (``defect''), and suppose that a defector does nothing and a cooperator incurs a cost, $w$, representing a fraction of his or her baseline reproductive capacity, $f_{N}$, in order to contribute to the provision of a public good. The public good is distributed evenly among all $N$ players in the population \citep{sigmund:PUP:2010}. Finally, a multiplication factor, $R>1$, quantifies the return on investment in this shared resource \citep{chen:PLOSONE:2012}.

The reproductive capacities of cooperators and defectors in state $\left(x_{C},x_{D}\right)$ are given by
\begin{linenomath}
\begin{subequations}\label{eq:cooperatorDefectorRate}
\begin{align}
F_{C}\left(x_{C},x_{D}\right) &= \left( 1-w + wR\frac{x_{C}}{x_{C}+x_{D}} \right) f_{x_{C}+x_{D}} ; \label{eq:cooperatorRate} \\
F_{D}\left(x_{C},x_{D}\right) &= \left( 1 + wR\frac{x_{C}}{x_{C}+x_{D}} \right) f_{x_{C}+x_{D}} . \label{eq:defectorRate}
\end{align}
\end{subequations}
\end{linenomath}
When $w\ll 1$, the contribution of this game to reproductive capacity is small. On the other hand, when $w=1$, cooperators expend their entire baseline reproductive capacity contributing to the public good. Unlike in many evolutionary games in populations of fixed size, where $w$ represents selection strength and quantifies relative differences between traits, here the cost of cooperation admits an intuitive biological interpretation: cooperators sacrifice a fraction $w$ of their expected number of offspring in hopes of seeing a return.

For a neutral population whose dynamics are governed only by the non-increasing baseline reproductive capacity, $f_{N}$, if $f_{1}<1$ then the time to extinction will be relatively short. In contrast, if $f_{1}>1$ the population will have a positive growth rate until $N$ becomes large enough that $f_{N}<1$. Then the population will grow to a stochastic carrying capacity and fluctuate around this size, possibly for considerable time. (For the two classes of baseline reproductive capacities we consider here, this carrying capacity need not be exactly $K$; see \ref{sec:appendixA}). We will refer to situations of this sort as ``metastable'' because all the populations we consider would eventually go extinct. According to \eq{cooperatorDefectorRate}, payoffs from the public goods game can increase reproductive capacities, with the possibility of positive population growth rates even if $f_{1}<1$. Due to our choice of non-increasing functions for $f_{N}$ (that decay to $0$ as $N$ grows), this will lead to metastable states but never to unbounded growth of the population. 

\section{Evolutionary dynamics of the Wright-Fisher branching process}
Let $\mathbf{E}_{\left(x_{C},x_{D}\right)}\left[y_{C}\right]$ and $\mathbf{E}_{\left(x_{C},x_{D}\right)}\left[y_{D}\right]$ denote the expected numbers of cooperator and defectors in the next generation given $x_{C}$ cooperators and $x_{D}$ defectors in the current generation. In this section, we are mainly interested in the existence of metastable equilibria, which are defined as states, $\left(x_{C}^{\ast},x_{D}^{\ast}\right)$, such that
\begin{linenomath}
\begin{subequations}
\begin{align}
\mathbf{E}_{\left(x_{C}^{\ast},x_{D}^{\ast}\right)}\left[y_{C}\right] &= x_{C}^{\ast} ; \\
\mathbf{E}_{\left(x_{C}^{\ast},x_{D}^{\ast}\right)}\left[y_{D}\right] &= x_{D}^{\ast} .
\end{align}
\end{subequations}
\end{linenomath}
Populations will fluctuate around these states for some time, although extinction is inevitable. The time to extinction depends on the population's carrying capacity (see below and \ref{sec:extinctionTime}). Even when the population eventually goes extinct with probability $1$, it can take extremely long to do so.

We are particularly interested in the case when a population of defectors cannot survive for long on their own but a population of cooperators can. While a population of defectors evolves according to $f_{N}$, a population of cooperators evolves according to the reproductive capacity $\left(1+w\left(R-1\right)\right) f_{N}$, which can be greater than $1$ even when $f_{N}<1$. In polymorphic populations, cooperators and defectors have reproductive capacities given by $F_{C}\left(x_{C},x_{D}\right)$ and $F_{D}\left(x_{C},x_{D}\right)$ in \eq{cooperatorDefectorRate}, which are functions of $x_{C}$ and $x_{D}$, but also depend on the baseline reproductive capacity, $f_{N}$, the cost of cooperation, $w$, and the multiplication factor for the public good, $R$. We also consider situations in which a population of defectors can reach a metastable carrying capacity, i.e. when $f_{1}>1$. In this case, we are interested in the effects that $w$ and $R$ can have on the numbers of cooperators and defectors in polymorphic metastable states.

\subsection{Selection dynamics (without mutation)}
When the initial numbers of cooperators and defectors are small, stochastic effects have a profound influence over the long-run composition of the population. We show in \ref{sec:appendixB} that any non-zero metastable equilibrium must be monomorphic (all-cooperator or all-defector) for the update rule defined by \eq{wfbUpdateRule}. Although defectors generally have larger growth rates than cooperators in mixed populations, they can go extinct quickly in small populations, which, in turn, can permit cooperators to prosper. For example, suppose that defectors cannot survive on their own ($f_{1}<1$), which means that any population of defectors shrinks, on average, from one generation to the next. If any population of cooperators grows, due to the multiplication factor $R>1$, then the only populations that persist beyond a short time horizon are those composed entirely of cooperators. Therefore, cooperators have a type of survivorship bias. \fig{smallInitialConditions} illustrates this phenomenon, showing that defectors often outcompete cooperators (approximately 85\% of the time) when both are in the population (A), but once one type goes extinct, the population must consist of just cooperators in order to survive for any considerable length of time (B). These simulations are done with the baseline capacity $f_{N}=\max\left\{0,c_{K}+r\left(1-\frac{N}{K}\right)\right\}$, where $c_{K}=0.99$, $r=0.01$, and $K=100$; cost of cooperation $w=0.1$; and multiplication factor $R=2.0$.

\begin{figure}
\centering
\includegraphics[width=0.8\textwidth]{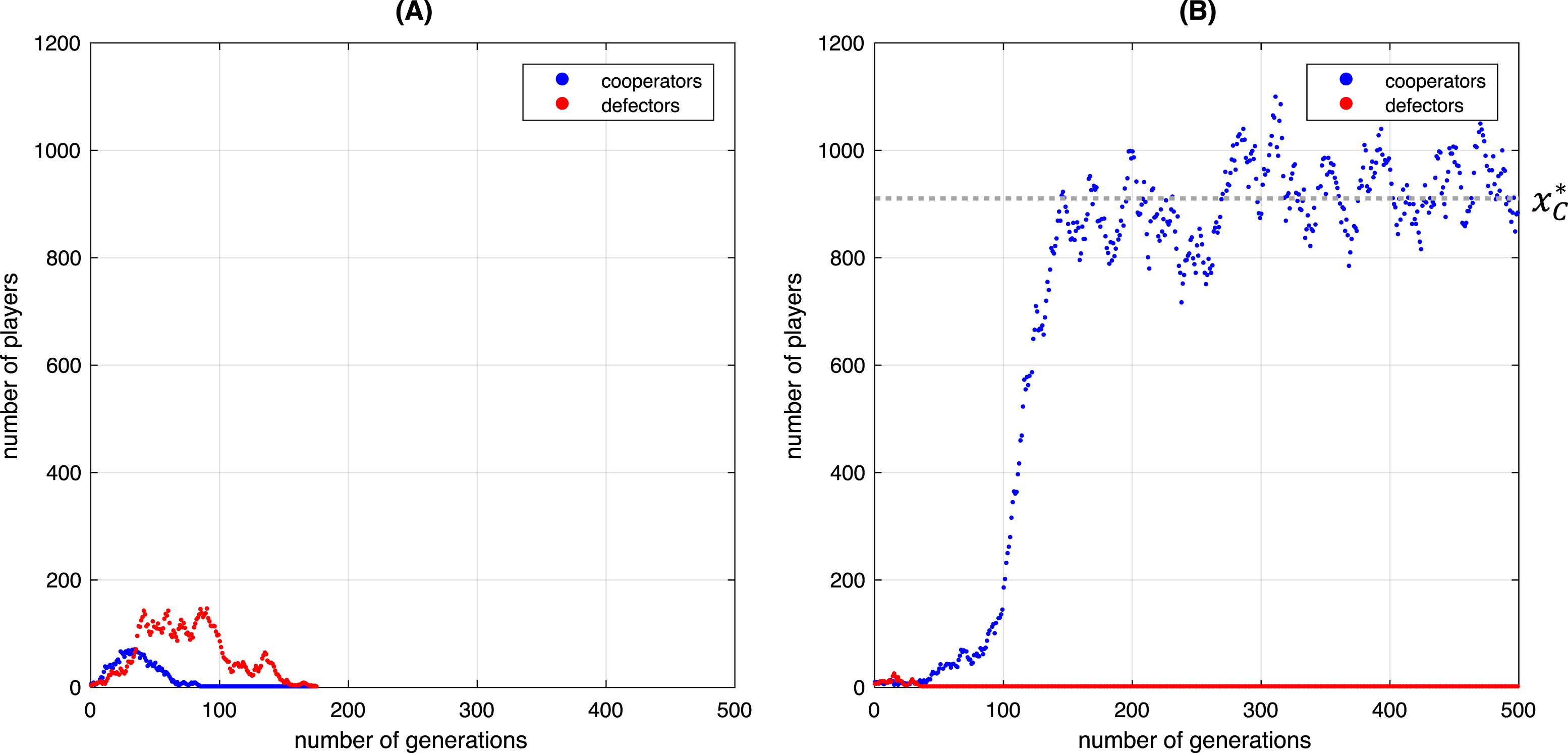}
\caption{Dynamics of the Wright-Fisher branching process in the absence of mutation. The baseline reproductive capacity is given by $f_{N}=\max\left\{0,c_{K}+r\left(1-\frac{N}{K}\right)\right\}$, where $c_{K}=0.99$, $r=0.01$, and $K=100$. The initial population consists of just $5$ cooperators and $5$ defectors. Approximately 15\% (rounded to the nearest percentage) of runs result in the behavior of panel (B), with defectors going extinct and cooperators reaching their carrying capacity. The remaining runs resemble panel (A), with almost-immediate extinction of the entire population. Without mutation, any metastable equilibrium is necessarily monomorphic; since defectors cannot survive on their own ($f_{1}<1$), it follows that only all-cooperator states can be observed as the long-run outcome of these initial conditions. Notably, while defectors go extinct in fewer than $200$ generations in (A), in (B) the population of cooperators thrives even after $10^{9}$ generations (although it goes extinct, eventually, with probability $1$). To demonstrate the initial ascent of cooperators, we include here only the first $500$ generations. Parameters: $u=0$, $w=0.1$, and $R=2.0$.\label{fig:smallInitialConditions}}
\end{figure}

\subsection{Mutation-selection dynamics}
A common way to quantify the evolutionary success of cooperators is to introduce strategy mutations and study the frequency of cooperators in the mutation-selection equilibrium \citep{tarnita:JTB:2009,antal:JTB:2009,tarnita:PNAS:2011}. Mutations indicate errors in the transmission (either cultural or genetic) of the two strategies (cooperation and defection) and can be small \citep{nowak:BP:2006,wu:JMB:2011} or large \citep{traulsen:PNAS:2009} depending on their interpretation. The success of cooperation is quantified by its average frequency in the population over many generations. In a population of cooperators and defectors under neutral drift (i.e. without selective differences between the two types), cooperators are indistinguishable from defectors and are equally frequent in the mutation-selection equilibrium. If selection brings the cooperator frequency above $1/2$, then selection is said to favor cooperation. By this metric, selection typically disfavors cooperation in unstructured populations \citep{tarnita:JTB:2009}.

If the population size is static and the update rule is that of the Wright-Fisher process, \eq{wfUpdateRule}, then the baseline reproductive capacity appearing in \eq{cooperatorDefectorRate}, $f_{N}$, cancels out. Only the relative fitnesses of cooperators and defectors matters. The dynamics are then captured in the relative frequencies of cooperators and defectors. Since cooperators are always less frequent than defectors when the intensity of selection, $w$, is positive, selection unambiguously disfavors cooperators relative to defectors. This result can be seen in \fig{staticDynamic}(D)--(F), in which results are shown for three different values of $w$. These simulations were generated using a multiplication factor of $R=2.0$ and a mutation rate of $u =0.01$. That selection favors defectors is a standard property of many social dilemmas in unstructured populations; additional mechanisms (such as spatial structure) must typically be present in order for cooperators to outperform defectors \citep{nowak:PTRSB:2009}.

When the population size can fluctuate and $u$ is the probability that a mutation occurs, then the dynamics are governed by \eq{wfbUpdateRuleMutation}. Here, it is still the case that selection decreases the frequency of cooperators relative to defectors. On the other hand, the population can quickly go extinct if selection is not sufficiently strong, which we illustrate in \fig{staticDynamic}(A)--(C) with $R=2.0$, $u=0.01$, and $f_{N} =\left(1+r\right)\min\left\{ 1 , \frac{K}{N} \right\}$ with $r=-0.001$ and $K=5000$. Thus, cooperation can be favored in such situations because it protects against extinction.

\begin{figure}
\centering
\includegraphics[width=0.8\textwidth]{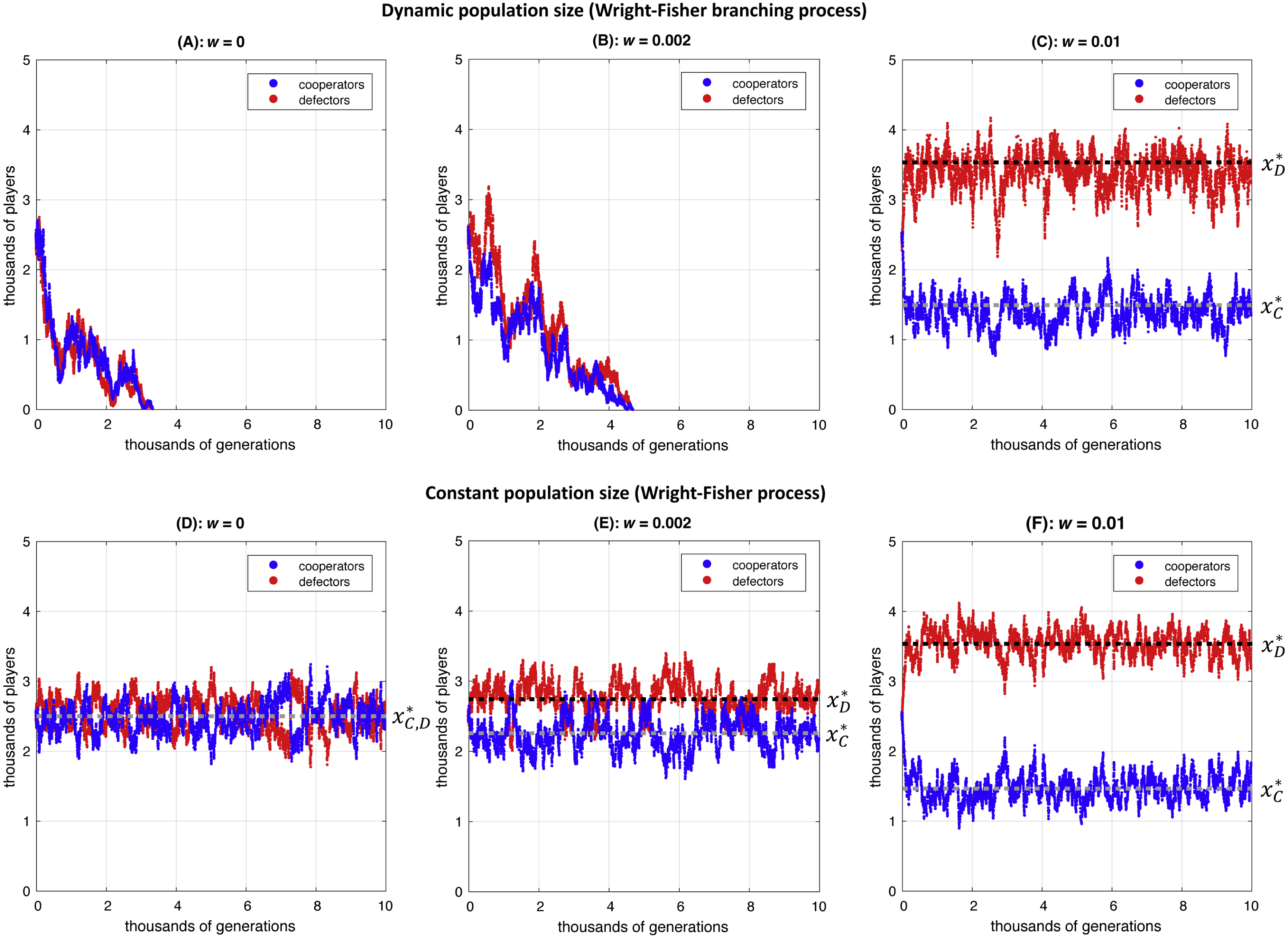}
\caption{Mutation-selection dynamics for the Wright-Fisher branching process (panels (A)--(C)) and the Wright-Fisher process with constant population size (panels (D)--(F)). The baseline reproductive capacity is $f_{N}=\left(1+r\right)\min\left\{1,\frac{K}{N}\right\}$, where $r=-0.001$ and $K=5000$, although this is relevant only for (A)--(C) since the population size is held fixed in (D)--(F). In all panels, the population initially consists of $2500$ cooperators and $2500$ defectors, the strategy mutation rate is $u =0.01$, and the multiplication factor for the public goods game is $R=2.0$. In the top row, the population size can fluctuate over time; in the bottom row, it is held constant at $5000$ players. In (A) and (D), there are no selective differences between cooperators and defectors ($w =0$). In (B) and (E), a cooperator sacrifices a small fraction, $w =0.002$, of his or her baseline reproductive capacity in order to provide the community with a benefit. In (C) and (F), cooperators sacrifice a larger portion, $w =0.01$, of their reproductive capacity when provisioning a public good. While the population is artificially prevented from going extinct in (D)--(F), it can go extinct in (A)--(C) and does so quickly when the cost of cooperation is too small (panels (A) and (B)) since $f_{1}<1$. Although increasing the cost of cooperation tends to decrease the frequency of cooperators relative to defectors, a sufficient amount of selection is necessary for the survival of the population as a whole. Therefore, there is an optimal cost of cooperation, $w^{\ast}$, falling between $0.002$ and $0.01$, that maximizes the frequency of cooperators subject to survival of the population.\label{fig:staticDynamic}}
\end{figure}

One key difference from models with fixed population size is that, in a branching process, the population either grows unboundedly or eventually goes extinct \citep{jagers:SPMS:2012,hamza:JMB:2015}. That is, if the population remains bounded in size, then the only true stationary state is extinction. Despite this behavior capturing the long-run dynamics of the process, there can also exist metastable equilibria in which the process persists prior to population extinction. We show in \ref{sec:extinctionTime} that the persistence time in our model grows exponentially in $K$ \citep[see also][]{jagers:JAP:2011,faure:AAP:2014,schreiber:S:2017}, meaning that if $\mathbf{E}\left[\tau_{K}\right]$ is the expected number of generations prior to extinction after starting in the quasi-stationary distribution, then there exists $c>0$ (independent of $K$) for which $\mathbf{E}\left[\tau_{K}\right]\geqslant e^{cK}$.

More informally, if $\left(x_{C}^{\ast},x_{D}^{\ast}\right)$ is a metastable equilibrium and $\sigma$ denotes standard deviation, then
\begin{linenomath}
\begin{subequations}
\begin{align}
\sigma_{\left(x_{C}^{\ast},x_{D}^{\ast}\right)}\left[y_{C}\right] / \mathbf{E}_{\left(x_{C}^{\ast},x_{D}^{\ast}\right)}\left[y_{C}\right] &= \frac{1}{\sqrt{x_{C}^{\ast}}} ; \\
\sigma_{\left(x_{C}^{\ast},x_{D}^{\ast}\right)}\left[y_{D}\right] / \mathbf{E}_{\left(x_{C}^{\ast},x_{D}^{\ast}\right)}\left[y_{D}\right] &= \frac{1}{\sqrt{x_{D}^{\ast}}} .
\end{align}
\end{subequations}
\end{linenomath}
Therefore, the fluctuations around a metastable equilibrium constitute only small fractions of $x_{C}^{\ast}$ and $x_{D}^{\ast}$ when $x_{C}^{\ast}$ and $x_{D}^{\ast}$ are sufficiently large (see \ref{sec:appendixB} for further details). Since $x_{C}^{\ast}$ and $x_{D}^{\ast}$ grow with $K$, and since the fluctuations in $x_{C}^{\ast}$ and $x_{D}^{\ast}$ are on the order of $\sqrt{x_{C}^{\ast}}$ and $\sqrt{x_{D}^{\ast}}$, respectively, the expected amount of time until deviations from the mean destroy the population, i.e. deviates so far as to hit $\left(x_C,x_D\right)=(0,0)$, grows rapidly in $K$. \fig{histograms} shows the quasi-stationary distributions of $x_{C}$ and $x_{D}$ that result from fluctuations around metastable equilibria, such as those shown in \fig{staticDynamic}(C).

\begin{figure}
\centering
\includegraphics[width=0.8\textwidth]{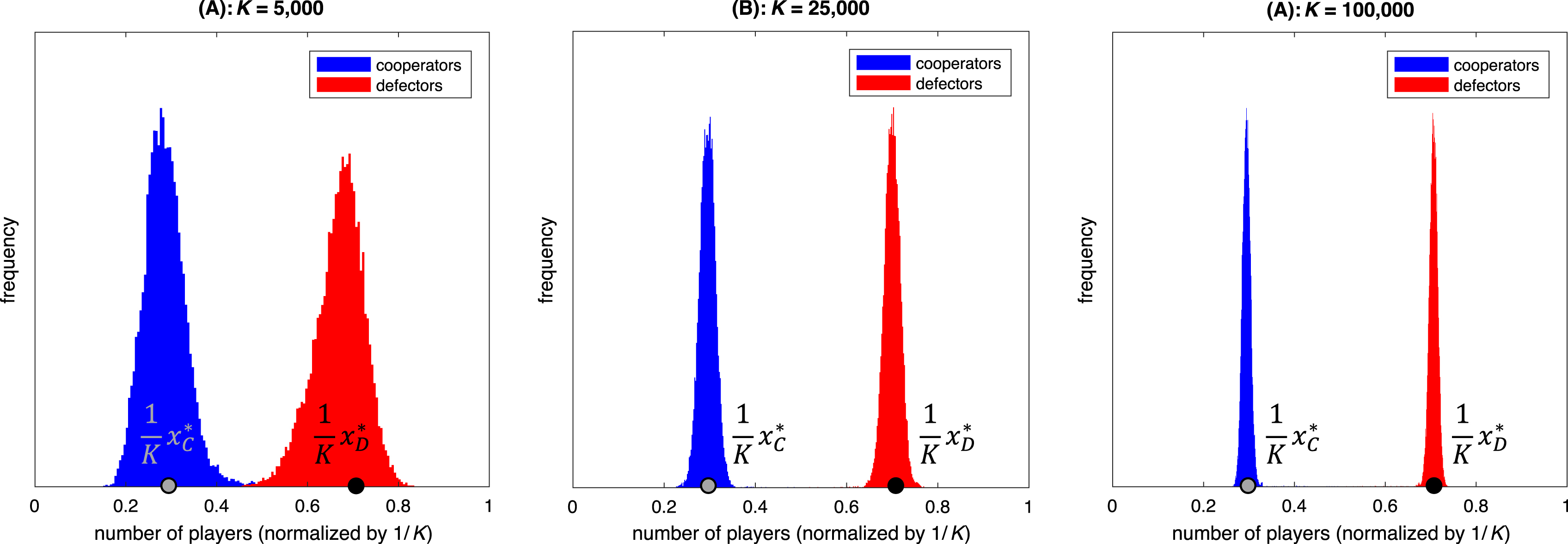}
\caption{Simulation of the quasi-stationary distribution for several values of $K$. Here, $K$ enters in the baseline reproductive capacity, $f_{N} =\left(1+r\right)\min\left\{1,\frac{K}{N}\right\}$, where $r=-0.001$. In each panel, cooperators and defectors are each initialized at an equal abundance of $K/2$. The plots are histograms for cooperator (blue) and defector (red) abundance over the first $25000$ generations. The equilibrium fraction of cooperators, $p$, depends on only $u$, $w$, and $R$ and is the same for all panels. Therefore, the peaks are centered at $x_{C}^{\ast}/K=0.999x\left(1+w\left(R-1\right) x\right)\approx 0.2954$ for cooperators and $x_{D}^{\ast}/K=0.999\left(1-x\right)\left(1+w\left(R-1\right) x\right)\approx 0.7066$ for defectors (see \ref{sec:appendixB}). As $K$ grows, this quasi-stationary distribution converges to the Dirac measure centered on $\left(0.2954,0.7066\right)$. Parameters: $u=w=0.01$ and $R=2.0$.\label{fig:histograms}}
\end{figure}

The dynamics of this public goods game result from the balance among three factors: mutation, selection, and population survival. Although long-term population survival can be achieved by increasing the cost of cooperation, $w$, it can also be destroyed by decreasing the mutation rate, $u$. In \ref{sec:appendixB}, we show that for any $N\geqslant 1$ and any non-zero mutation rate and cost of cooperation, there exists a critical multiplication factor, $R_{N}^{\ast}$, such that the population is supported at a metastable equilibrium consisting of at least $N$ players whenever $R>R_{N}^{\ast}$. In general, the harmful effect (population extinction) of either low costs of cooperation or low mutation rates can be mitigated by increasing the return on investment in the public good.

Selection can also increase cooperator abundance while decreasing their relative frequency (\fig{cooperatorCount}). This phenomenon is a consequence of the fact that the presence of cooperators can change the carrying capacity of population. That abundance and frequency can move in opposite directions is unique to models with variable population size and presents an interesting question about the definition of cooperator success. We show in \ref{sec:appendixB} that the fraction of cooperators present in a metastable equilibrium is independent of $f_{N}$ and depends on just $u$, $w$, and $R$. Thus, when $u$, $w$, and $R$ are fixed, defectors claim a fixed fraction (at least $1/2$) of the total population, which means that cooperators are disfavored relative to defectors. However, based on population growth alone, cooperators could be considered to be favored by selection in an absolute sense because their abundance is an increasing function of the cost of cooperation, $w$.

\begin{figure}
\centering
\includegraphics[width=0.8\textwidth]{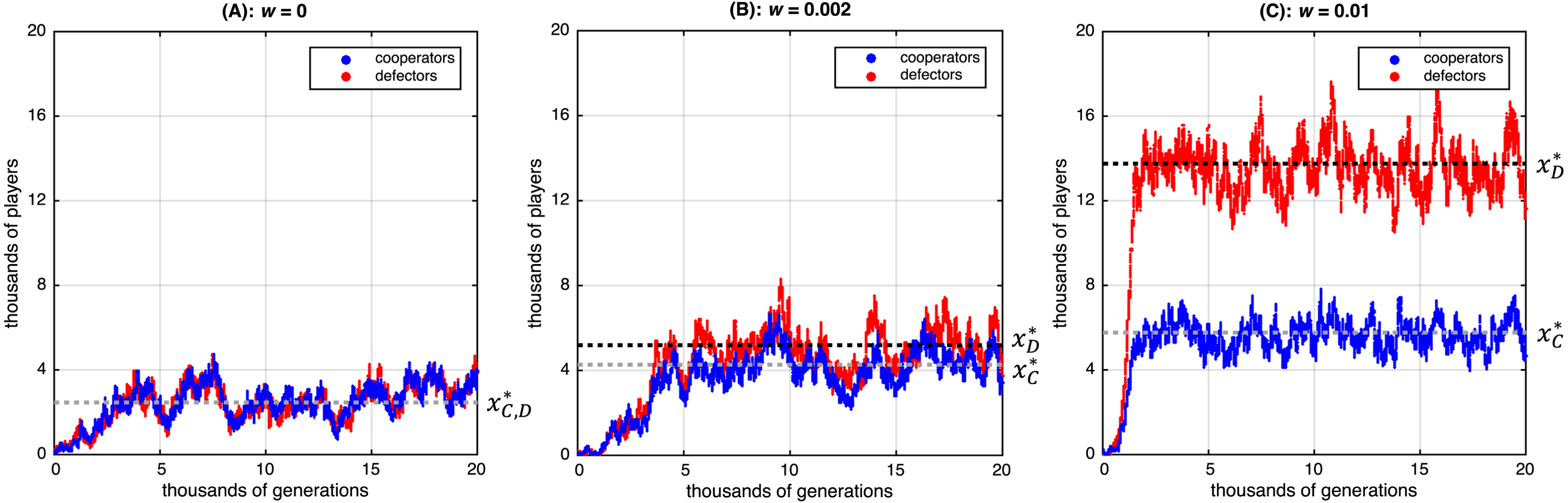}
\caption{Mutation-selection dynamics when a population of defectors can sustain itself at a positive carrying capacity ($f_{1}>1$). The baseline reproductive capacity is $f_{N}=\max\left\{0,c_{K}+r\left(1-\frac{N}{K}\right)\right\}$, where $c_{K}=1$, $r=0.001$, and $K=5000$. In each panel, the population starts out with $100$ cooperators and $100$ defectors. In (A), $w=0$ and cooperators are indistinguishable from defectors. The population grows to a metastable equilibrium with roughly equal frequencies of cooperators and defectors; $\left(x_{C}^{\ast},x_{D}^{\ast}\right) =\left(2500,2500\right)$. In (B), the cost of cooperation is positive ($w =0.002$) and defectors begin to outnumber cooperators; $\left(x_{C}^{\ast},x_{D}^{\ast}\right)\approx\left(4287,5219\right)$. At the metastable equilibrium, however, cooperators in (B) slightly outnumber cooperators in (A). In (C), the cost of cooperation is increased to $w =0.01$ and the gap between cooperator and defector abundance widens; $\left(x_{C}^{\ast},x_{D}^{\ast}\right)\approx\left(5806,13890\right)$. Although cooperators are less frequent than defectors in (C), they are more abundant in (C) than they were in (A) and (B), suggesting that their abundance is favored by selection even though their relative frequency is not. Parameters: $u=0.01$ and $R=2.0$.\label{fig:cooperatorCount}}
\end{figure}

\section{Discussion}
Public goods games have been used to model conflicts of interest ranging from cooperation in microorganisms \citep{craigmaclean:JEB:2008,czaran:PLOSONE:2009,cordero:PNAS:2012,sanchez:PLOSB:2013,allen:EL:2013}, to alarm calls in monkeys \citep{seyfarth:S:1980,cluttonbrock:W:2016}, sentinel behavior in meerkats \citep{cluttonbrock:Science:1999}, and large-scale human efforts aimed at combating climate change \citep{milinski:PNAS:2006,jacquet:NCC:2013} and pollution \citep{ehmke:EDE:2008}. Due to its linearity and close relationship to the prisoner's dilemma, the public goods game we consider is sometimes called the ``$N$-person prisoner's dilemma'' \citep{archetti:JTB:2012}. Provided $1\leqslant R\leqslant N$, this game presents a conflict of interest between the group and the individual that can be reduced to a sequence of $N-1$ prisoner's dilemma interactions \citep{hauert:Complexity:2003}. However, the analysis and interpretation of a single public goods game is somewhat more straightforward than that of a series of prisoner's dilemma interactions when the population size fluctuates over time.

In populations of fixed size, extinction is impossible and defectors can survive without the support of cooperators. This point marks perhaps the most prominent feature of classical models in evolutionary game theory that breaks down when the population size can fluctuate over time. When populations vary in size and defectors cannot sustain themselves on their own, cooperators must be present and selection must be sufficiently strong in order to maintain the existence of the population (\fig{staticDynamic}). Furthermore, when the population size is assumed to be fixed, selection decreases the frequency of cooperators if and only if it decreases the number, or abundance, of cooperators. In fluctuating populations, selection can decrease the frequency of cooperators while increasing their abundance (\fig{cooperatorCount}).

We have referred to $w$ as the ``cost of cooperation'' because of its interpretation as the expected fraction of offspring that must be sacrificed in order to cooperate. However, we note that because this fraction of an individual's baseline reproductive capacity is shared across the population, larger $w$ also means a greater effect of cooperation (similar to the return, $R$). Further, if $w=0$ then the population is identical to a population of defectors. In standard Wright-Fisher processes, $f_{N}$ is irrelevant and it is common to rewrite the terms in \eq{cooperatorDefectorRate} as $F_{C}=1+w\pi_{C}$ and $F_{D}=1+w\pi_{D}$, and to refer to $w$ the ``strength'' \citep{antal:JTB:2009,chastain:PNAS:2014,allen:Nature:2017} or ``intensity'' of selection \citep{nowak:Nature:2004,wu:PRE:2010,wu:PLoSCB:2013}. When the population size is held constant in our model, $w$ corresponds exactly to this well-known notion of selection intensity.

That the frequencies of cooperators relative to defectors in the metastable equilibria of \textbf{Figs. \ref{fig:staticDynamic}}(C), \textbf{\ref{fig:staticDynamic}}(F), and \textbf{\ref{fig:cooperatorCount}}(C) are all the same is not a coincidence. The fraction of cooperators, $p$, present at a metastable equilibrium is independent of the baseline reproductive capacity, $f_{N}$, and depends on only the mutation rate, $u$; the cost of cooperation, $w$; and the multiplication factor for the public goods game, $R$. The population size at a metastable equilibrium, however, does depend on $f_{N}$. In \ref{sec:appendixB}, we give an explicit formula for $p$ and a condition for the existence of a non-zero metastable equilibrium in terms of $f_{N}$, $w$, $R$, and $p$.

In the absence of mutation, either cooperators or defectors must be extinct in any metastable equilibrium. Although defectors outperform cooperators in a mixed population, a population of cooperators reaches a higher carrying capacity and persists at this size for a longer time than does a population of defectors. Small populations of cooperators have a distinct advantage over their all-defector counterparts due to larger growth rates. In particular, quick extinction is less likely for all-cooperator populations than it is for all-defector populations, reflecting observations of \citet{huang:PNAS:2015} and \citet{waite:PLOSCB:2015} for related models.

Unlike in the models of \citet{houchmandzadeh:BMC:2012} and \citet{houchmandzadeh:B:2015}, the population size in our model is not a deterministic function of the fraction of cooperators. Rather, it is a random quantity derived from the collective offspring pool of the parental generation. The population size at time $t+1$ depends on both the number of cooperators \textit{and} the number of defectors at time $t$. 

A framework more similar to ours is that of \citet{behar:PB:2016}, which uses stochastic differential equations to model the numbers of producers and non-producers of a common resource. Both numbers increase when small but eventually non-producers drive the population to extinction. Analogous to the possible role of mutations described here, \citet{behar:PB:2016} allow migration to reseed populations with producers. A metastable equilibrium may then occur in the total population even as each local population experiences boom and bust cycles. Our focus here has been on treating baseline reproductive dynamics as an exogenous feature and understanding on how these may be perturbed by a game to allow a variety of different carrying capacities to emerge depending on the parameters of the model.

Since population size can fluctuate in our model, one could also allow the multiplication factor of the public good, represented here by $R$, to change with $N$. If this multiplication factor gets weaker as $N$ grows, then one observes dynamics similar to those here even if $f_{N}$ is independent of $N$. Viewing $R$ as a function of $N$ presents an alternative way to model populations that cannot have unbounded growth due to environmental constraints. Another extension of our model could involve asymmetric mutation rates with, for example, $C\rightarrow D$ mutations more likely than $D\rightarrow C$. Although the importance of asymmetric mutation in population models is well-established \citep{eigen:TJPC:1988,durrett:G:2008,arnoldt:JRSI:2012}, we do not expect this would cause any qualitative changes in the results reported here unless the asymmetry was very extreme. 

Incorporating dynamic population size into classical evolutionary models complicates the analysis of their dynamics. Notably, how one measures the evolutionary success of cooperators is not as unambiguous here as it is in models with fixed population size. We have shown that selection can favor cooperator abundance despite disfavoring cooperator frequency, and that even though cooperators are exploited by defectors, they can be crucial to the survival of the population as a whole.

\setcounter{section}{0}
\renewcommand{\thesection}{Appendix~\Alph{section}}
\renewcommand{\thesubsection}{\Alph{section}.\arabic{subsection}}
\renewcommand{\thesubsubsection}{\Alph{section}.\arabic{subsection}.\arabic{subsubsection}}

\section{Wright-Fisher branching process}\label{sec:appendixA}
\subsection{Update rule}
Suppose that the population is unstructured but allowed to vary in size. For simplicity, assume that we are dealing with a symmetric game with two strategies, $C$ (``cooperate'') and $D$ (``defect''). A state of the population is then uniquely defined by a pair, $\left(x_{C},x_{D}\right)$, where $x_{C}$ and $x_{C}$ are the number of players using $C$ and $D$, respectively. The population size is $N=x_{C}+x_{D}$, which can vary over time.

Suppose that the reproductive capacities of cooperators and defectors in state $\left(x_{C},x_{D}\right)$ are given by functions $F_{C}=F_{C}\left(x_{C},x_{D}\right)$ and $F_{D}=F_{D}\left(x_{C},x_{D}\right)$, respectively. That is, the reproductive capacities are frequency-dependent and determined by the number of each type of player in the population. We define reproductive capacity as the expected number of surviving offspring of an individual over its lifetime. Ours is therefore an ``absolute'' interpretation of reproductive capacity. We assume a reproductive mechanism in which the number of offspring is Poisson-distributed with mean equal to the parent's reproductive capacity. Therefore, the probability of transitioning from $\left(x_{C},x_{D}\right)$ to $\left(y_{C},y_{D}\right)$ over a single generation in this ``Wright-Fisher branching process'' is 
\begin{linenomath}
\begin{align}\label{eq:updateRule}
\mathbf{P}_{\textrm{WFB}}(y_{C},y_{D} \mid x_{C},x_{D}) &= \left(\frac{\left(x_{C}F_{C}\right)^{y_{C}}e^{-x_{C}F_{C}}}{y_{C}!}\right)\left(\frac{\left(x_{D}F_{D}\right)^{y_{D}}e^{-x_{D}F_{D}}}{y_{D}!}\right) .
\end{align}
\end{linenomath}
As in the standard Wright-Fisher process, we assume that generations are non-overlapping.

\begin{remark}
If the population size is static and fixed at $N$, then, for $x_{C}+x_{D}=y_{C}+y_{D}=N$,
\begin{linenomath}
\begin{align}
\mathbf{P}_{\textrm{WFB}}&(y_{C},y_{D} \mid x_{C},x_{D}\ ;\ \textrm{fixed population size} ) \nonumber \\
&= \frac{\left(\frac{\left(x_{C}F_{C}\right)^{y_{C}}e^{-x_{C}F_{C}}}{y_{C}!}\right)\left(\frac{\left(x_{D}F_{D}\right)^{y_{D}}e^{-x_{D}F_{D}}}{y_{D}!}\right)}{\sum_{z_{C}+z_{D}=N}\left(\frac{\left(x_{C}F_{C}\right)^{z_{C}}e^{-x_{C}F_{C}}}{z_{C}!}\right)\left(\frac{\left(x_{D}F_{D}\right)^{z_{D}}e^{-x_{D}F_{D}}}{z_{D}!}\right)} \nonumber \\
&= \frac{\left(\frac{\left(x_{C}F_{C}\right)^{y_{C}}e^{-x_{C}F_{C}}}{y_{C}!}\right)\left(\frac{\left(x_{D}F_{D}\right)^{y_{D}}e^{-x_{D}F_{D}}}{y_{D}!}\right)}{\frac{1}{N!}\sum_{z_{C}+z_{D}=N}N!\left(\frac{\left(x_{C}F_{C}\right)^{z_{C}}e^{-x_{C}F_{C}}}{z_{C}!}\right)\left(\frac{\left(x_{D}F_{D}\right)^{z_{D}}e^{-x_{D}F_{D}}}{z_{D}!}\right)} \nonumber \\
&= \frac{\left(\frac{\left(x_{C}F_{C}\right)^{y_{C}}e^{-x_{C}F_{C}}}{y_{C}!}\right)\left(\frac{\left(x_{D}F_{D}\right)^{y_{D}}e^{-x_{D}F_{D}}}{y_{D}!}\right)}{\frac{1}{N!}\left(x_{C}F_{C}+x_{D}F_{D}\right)^{N}e^{-\left(x_{C}F_{C}+x_{D}F_{D}\right)}} \nonumber \\
&= \frac{N!}{y_{C}!y_{D}!}\left(\frac{x_{C}F_{C}}{x_{C}F_{C}+x_{D}F_{D}}\right)^{y_{C}}\left(\frac{x_{D}F_{D}}{x_{C}F_{C}+x_{D}F_{D}}\right)^{y_{D}} \nonumber \\
&= \mathbf{P}_{\textrm{WF}}(y_{C} \mid x_{C}) ,
\end{align}
\end{linenomath}
recovering the classical transition rule based on binomial sampling \citep{haccou:CUP:2005}. Therefore, the update rule defined by \eq{updateRule} may be thought of as a generalization of the classical, frequency-dependent Wright-Fisher process.
\end{remark}

\subsection{Reproductive capacities and selection}
In the absence of selection, each player in a population of size $N$ has a reproductive capacity determined by a baseline reproductive capacity, $f_{N}$. We consider the following two functional forms for $f_{N}$, examples of which are depicted in \fig{baselineRates} in the main text.

\subsubsection{Rectified linear}
One natural way to model reproductive capacity is as a linear function of the population size, $N$. In this case, we can write $f_{N} =\max\left\{ 0 , c_{K}+r\left(1-\frac{N}{K}\right) \right\}$ for some parameters $c_{K}$, $r$, and $K$. We refer to $f_{N}$ as a ``rectified'' linear reproductive capacity since it piecewise-linear with the constraint $f_{N}\geqslant 0$ for every $N$. Note that $f_{N^{\ast}}=1$ when $N^{\ast}=K\left(1-\frac{1}{r}\left(1-c_{K}\right)\right)$. Therefore, $N^{\ast}$ may be interpreted as the (neutral) carrying capacity of the population since when $f_{N^{\ast}}=1$ each individual is replaced by one offspring on average. Note that $K$ itself is not necessarily the neutral carrying capacity for this form of $f_{N}$.

\subsubsection{Threshold-constant}
If the reproductive capacity is constant, then every player expects to produce $1+r$ offspring that survive into the next generation, where $r\geqslant -1$. We assume that this growth is eventually bounded by environmental constraints, so we set $f_{N} =\left(1+r\right)\min\left\{1,K/N\right\}$ for some $K$. We refer to $f_{N}$ as a ``threshold-constant'' reproductive capacity since it is constant up to a threshold ($N=K$) and then decreasing to $0$ beyond $K$. When $r<0$, there is no solution to $f_{N}=1$ since $f_{N}\leqslant 1+r<1$ for each $N$. When $r>0$, we have $f_{\left(1+r\right) K}=1$, so $N^{\ast}=\left(1+r\right) K$ is the neutral carrying capacity of the population.

\subsubsection{Selection}
In a game with strategies $C$ and $D$, let $\pi_{C}\left(x_{C},y_{C}\right)$ and $\pi_{D}\left(x_{C},x_{D}\right)$ be the total payoffs to $C$ and $D$, respectively, when there are $x_{C}$ cooperators and $x_{D}$ defectors. If the population size is fixed, then a payoff of $\pi$ is typically converted to a fitness of $f$ by defining $f=1+w\pi$, where $w$ is the ``selection strength'' \citep[see][]{antal:JTB:2009,antal:PNAS:2009}. This perturbation approach has even been extended to asymmetric games played between different populations \citep{veller:JET:2016}. While our setup is somewhat different, we maintain this convention of using payoffs from a game to perturb reproductive capacities. In particular, if $w$ is a parameter representing the intensity of selection, then the reproductive capacities of cooperators and defectors are given by
\begin{linenomath}
\begin{subequations}
\begin{align}\label{eq:rateFunction}
F_{C}\left(x_{C},x_{D}\right) &:= \Big( 1 + w \pi_{C}\left(x_{C},x_{D}\right) \Big) f_{x_{C}+x_{D}} ; \\
F_{D}\left(x_{C},x_{D}\right) &:= \Big( 1 + w \pi_{D}\left(x_{C},x_{D}\right) \Big) f_{x_{C}+x_{D}} ,
\end{align}
\end{subequations}
\end{linenomath}
respectively. In other words, the baseline reproductive capacity, $f_{x_{C}+x_{D}}$, is perturbed by the game according to the strength of selection, $w$. In order to maintain non-negative reproductive capacities, $w$ must be sufficiently small. In the next section, we consider a public goods game in which $w$ has a clear biological interpretation.

\section{Dynamics of the public goods game}\label{sec:appendixB}
In the public goods game, a cooperator sacrifices a fraction, $w$, of his or her baseline reproductive capacity in order to contribute to a public good. This contribution is enhanced by a factor of $R>1$ and then distributed evenly among all of the players in the population. In terms of the payoff function in \eq{rateFunction}, we have $\pi_{C}\left(x_{C},x_{D}\right) =R\left(\frac{x_{C}}{x_{C}+x_{D}}\right) - 1$ and $\pi_{D}\left(x_{C},x_{D}\right) =R\left(\frac{x_{C}}{x_{C}+x_{D}}\right)$ as well as \eq{cooperatorDefectorRate} in the main text.

\subsection{Metastable equilibria}
Consider a population evolving according to the update rule of \eq{wfbUpdateRuleMutation}. As noted in the main text, such a branching process either grows without bound or eventually goes extinct. Even when the population has an extinction probability of $1$, there can be so-called ``metastable'' states (or ``equilibria'') around which the population fluctuates for many generations. While a quasi-stationary distribution for the process describes the distribution of strategy abundances prior to extinction, a metastable equilibrium describes the mean(s) around which these strategy counts fluctuate. We are interested in when these metastable equilibria exist and how they are influenced by the parameters of the model.

Let $\mathbf{E}_{\left(x_{C},x_{D}\right)}\left[y_{C}\right]$ (resp. $\mathbf{E}_{\left(x_{C},x_{D}\right)}\left[y_{D}\right]$) be the expected abundance of cooperators (resp. defectors) in the next generation given $x_{C}$ cooperators and $x_{D}$ defectors in the current generation. Formally, a metastable equilibrium for this process is a state at which $\mathbf{E}_{\left(x_{C}^{\ast},x_{D}^{\ast}\right)}\left[y_{C}\right] =x_{C}^{\ast}$ and $\mathbf{E}_{\left(x_{C}^{\ast},x_{D}^{\ast}\right)}\left[y_{D}\right] =x_{D}^{\ast}$. That is, cooperator and defector abundances each remain unchanged (on average) at a metastable equilibrium. We use the term ``metastable'' because the population fluctuates around this state but eventually goes extinct. We discuss extinction time in \ref{sec:extinctionTime}. First, we derive the metastable equilibria for public goods games.

\subsubsection{Derivation of metastable equilibria}
Let $u$ be the strategy-mutation rate. With probability $1-u$, an offspring acquires the strategy of the parent. With probability $u$, the offspring takes on one of $C$ and $D$ uniformly at random. In state $\left(x_{C},x_{D}\right)$, the expected number of cooperators in the next generation is
\begin{linenomath}
\begin{align}
\mathbf{E}_{\left(x_{C},x_{D}\right)}\left[y_{C}\right] &= \sum_{\left(y_{C},y_{D}\right)}\mathbf{P}_{\textrm{WFB}}(y_{C},y_{D} \mid x_{C},x_{D})\left(\left(1-\frac{u}{2}\right)y_{C}+\left(\frac{u}{2}\right)y_{D}\right) \nonumber \\
&= \sum_{\left(y_{C},y_{D}\right)}\left(\frac{\left(x_{C}F_{C}\right)^{y_{C}}e^{-x_{C}F_{C}}}{y_{C}!}\right)\left(\frac{\left(x_{D}F_{D}\right)^{y_{D}}e^{-x_{D}F_{D}}}{y_{D}!}\right)\left(\left(1-\frac{u}{2}\right)y_{C}+\left(\frac{u}{2}\right)y_{D}\right) \nonumber \\
&= \left(1-\frac{u}{2}\right)x_{C}F_{C} + \left(\frac{u}{2}\right)x_{D}F_{D} .
\end{align}
\end{linenomath}
Similarly, the expected number of defectors in the next generation is $\left(\frac{u}{2}\right)x_{C}F_{C} + \left(1-\frac{u}{2}\right)x_{D}F_{D}$. Therefore, the system of equations we need to solve in order to find a metastable equilibrium is
\begin{linenomath}
\begin{subequations}\label{eq:alphaCalphaDfCfD}
\begin{align}
x_{C} &= \left(1-\frac{u}{2}\right)x_{C}F_{C} + \left(\frac{u}{2}\right)x_{D}F_{D} ; \label{eq:alphaCfCfD} \\
x_{D} &= \left(\frac{u}{2}\right)x_{C}F_{C} + \left(1-\frac{u}{2}\right)x_{D}F_{D} . \label{eq:alphaDfCfD}
\end{align}
\end{subequations}
\end{linenomath}
In other words, it must be true that
\begin{linenomath}
\begin{subequations}\label{eq:alphaCalphaD}
\begin{align}
x_{C} &= \left[\left(1-\frac{u}{2}\right)x_{C}\Big(1 +w \pi_{C}\left(x_{C},x_{D}\right)\Big) + \left(\frac{u}{2}\right)x_{D}\Big(1 +w \pi_{D}\left(x_{C},x_{D}\right)\Big)\right] f_{x_{C}+x_{D}} ; \label{eq:alphaC} \\
x_{D} &= \left[\left(\frac{u}{2}\right)x_{C}\Big(1 +w \pi_{C}\left(x_{C},x_{D}\right)\Big) + \left(1-\frac{u}{2}\right)x_{D}\Big(1 +w \pi_{D}\left(x_{C},x_{D}\right)\Big)\right] f_{x_{C}+x_{D}} . \label{eq:alphaD}
\end{align}
\end{subequations}
\end{linenomath}
These equations are trivially satisfied when $x_{C}=x_{D}=0$ (population extinction). There can also be solutions to \eq{alphaCalphaD} with $x_{C}\neq 0$ or $x_{D}\neq 0$; we give a condition for the existence of non-zero solutions below.

\begin{remark}
If $u =0$, then \eq{alphaCalphaD} reduces to the system
\begin{linenomath}
\begin{subequations}
\begin{align}
x_{C} &= \Big(1 + w \pi_{C}\left(x_{C},x_{D}\right)\Big) x_{C} f_{x_{C}+x_{D}} ; \\
x_{D} &= \Big(1 + w \pi_{D}\left(x_{C},x_{D}\right)\Big) x_{D} f_{x_{C}+x_{D}} .
\end{align}
\end{subequations}
\end{linenomath}
If $x_{C}$ and $x_{D}$ satisfy this system and $x_{C},x_{D}\neq 0$, then
\begin{linenomath}
\begin{align}
\Big(1 + w \pi_{C}\left(x_{C},x_{D}\right)\Big) f_{x_{C}+x_{D}} = \Big(1 + w \pi_{D}\left(x_{C},x_{D}\right)\Big) f_{x_{C}+x_{D}} = 1 .
\end{align}
\end{linenomath}
Therefore, either $w=0$ and $f_{x_{C}+x_{D}} =1$ or $w\neq 0$ and $\pi_{C}\left(x_{C},x_{D}\right) =\pi_{D}\left(x_{C},x_{D}\right)$. However, for the public goods game, it is always the case that $\pi_{C}\left(x_{C},x_{D}\right) <\pi_{D}\left(x_{C},x_{D}\right)$ when $x_{C}>0$, so it must be true that $w=0$ and $f_{x_{C}+x_{D}} =1$. Thus, if $u =0$ and $w>0$, then any solution satisfies $x_{C}=0$ or $x_{D}=0$. In other words, in the absence of mutation, selection forces the extinction of at least one strategy.

Similarly, if $u>0$ and there is a solution $\left(x_{C},x_{D}\right)$ with $x_{C}=0$, then one obtains $x_{D}F_{D}=0$ from \eq{alphaCfCfD} and $x_{D}=0$ from \eq{alphaDfCfD}. For a similar reason, if there is solution with $x_{D}=0$, then it must also be true that $x_{C}=0$. Thus, if $\left(x_{C},x_{D}\right)$ is a non-zero solution to \eq{alphaCalphaD} when $u>0$, then $x_{C}>0$ and $x_{D}>0$.
\end{remark}

\begin{lemma}\label{lem:cooperatorFraction}
If $u>0$, then there exists $p\in\left(0,1\right)$ such that any non-zero solution to \eq{alphaCalphaD}, $\left(x_{C},x_{D}\right)$, satisfies $\frac{x_{C}}{x_{C}+x_{D}} = p$. In other words, the fraction of cooperators is the same in any solution to \eq{alphaCalphaD}.
\end{lemma}
\begin{proof}
If $\left(x_{C},x_{D}\right)$ is a solution to \eq{alphaCalphaD} with $x_{C},x_{D}\geqslant 0$ and $x_{C}+x_{D}>0$, then, with $p:=\frac{x_{C}}{x_{C}+x_{D}}$,
\begin{linenomath}
\begin{align}
p &= \frac{x_{C}}{x_{C}+x_{D}} \nonumber \\
&= \frac{\left(1-\frac{u}{2}\right)x_{C}\Big(1 +w \pi_{C}\left(x_{C},x_{D}\right)\Big) + \left(\frac{u}{2}\right)x_{D}\Big(1 +w \pi_{D}\left(x_{C},x_{D}\right)\Big)}{x_{C}\Big(1 +w \pi_{C}\left(x_{C},x_{D}\right)\Big) + x_{D}\Big(1 +w \pi_{D}\left(x_{C},x_{D}\right)\Big)} \nonumber \\
&= \frac{\left(1-\frac{u}{2}\right)\left(\frac{x_{C}}{x_{C}+x_{D}}\right)\Big(1 +w \pi_{C}\left(\frac{x_{C}}{x_{C}+x_{D}},\frac{x_{D}}{x_{C}+x_{D}}\right)\Big) + \left(\frac{u}{2}\right)\left(\frac{x_{D}}{x_{C}+x_{D}}\right)\Big(1 +w \pi_{D}\left(\frac{x_{C}}{x_{C}+x_{D}},\frac{x_{D}}{x_{C}+x_{D}}\right)\Big)}{\left(\frac{x_{C}}{x_{C}+x_{D}}\right)\Big(1 +w \pi_{C}\left(\frac{x_{C}}{x_{C}+x_{D}},\frac{x_{D}}{x_{C}+x_{D}}\right)\Big) + \left(\frac{x_{D}}{x_{C}+x_{D}}\right)\Big(1 +w \pi_{D}\left(\frac{x_{C}}{x_{C}+x_{D}},\frac{x_{D}}{x_{C}+x_{D}}\right)\Big)} \nonumber \\
&= \frac{\left(1-\frac{u}{2}\right) p\Big(1 +w \pi_{C}\left(p,1-p\right)\Big) + \left(\frac{u}{2}\right)\left(1-p\right)\Big(1 +w \pi_{D}\left(p,1-p\right)\Big)}{p\Big(1 +w \pi_{C}\left(p,1-p\right)\Big) + \left(1-p\right)\Big(1 +w \pi_{D}\left(p,1-p\right)\Big)} . \label{eq:fractionEquation}
\end{align}
\end{linenomath}
Since $\pi_{C}\left(x_{C},x_{D}\right) =R\left(\frac{x_{C}}{x_{C}+x_{D}}\right) - 1$ and $\pi_{D}\left(x_{C},x_{D}\right) =R\left(\frac{x_{C}}{x_{C}+x_{D}}\right)$, \eq{fractionEquation} is equivalent to
\begin{linenomath}
\begin{align}\label{eq:quadraticPolynomial}
\varphi\left(p\right) &:= 2\left(1-u R\right) wp^{2} + \left(u w\left(R+1\right) -2u -2w\right) p + u = 0 .
\end{align}
\end{linenomath}
Since $\varphi\left(p\right)$ is (at most) quadratic, $\varphi\left(0\right) =u$, and $\varphi\left(1\right) =-u\left(1-w+Rw\right)$, we see that if $u >0$, then there is a unique solution to \eq{fractionEquation} that falls within $\left[0,1\right]$, and, furthermore, this solution is in $\left(0,1\right)$. Explicitly,
\begin{linenomath}
\begin{align}\label{eq:equilibriumFraction}
p &= \frac{u w\left(R+1\right) -2u -2w +\sqrt{u^{2}w^{2}R^{2}+2u^{2}w^{2}R+4u^{2}wR-4u w^{2}R+u^{2}w^{2}-4u^{2}w+4u^{2}-4u w^{2}+4w^{2}}}{4\left(u R -1\right) w}
\end{align}
\end{linenomath}
if $R\neq 1/u$ and $w\neq 0$, and $p=u /\left(2u +\left(1-u\right) w\right)$ if $R=1/u$ or $w=0$, which completes the proof.
\end{proof}

From the proof of Lemma \ref{lem:cooperatorFraction}, we see that if $u =0$, then either (i) $w=0$ and every $p\in\left[0,1\right]$ is a solution to \eq{fractionEquation} or (ii) $w>0$ and the only solutions to \eq{fractionEquation} are $p=0$ and $p=1$. For any $u\in\left[0,1\right]$, the fraction of cooperators in a non-zero metastable equilibrium is independent of the baseline reproductive capacity, $f_{N}$. However, the existence of a metastable equilibrium and the size of the population at such an equilibrium both depend on the baseline reproductive capacity. Suppose that $x_{C}=pN$ and $x_{D}=\left(1-p\right) N$ satisfy \eq{alphaCalphaD}, where $p\in\left[0,1\right]$ is a fraction of cooperators that satisfies \eq{fractionEquation}. From \eq{alphaCalphaD},
\begin{linenomath}
\begin{subequations}\label{eq:xCondition}
\begin{align}
p &= \left[\left(1-\frac{u}{2}\right) p\Big(1 +w \pi_{C}\left(p,1-p\right)\Big) + \left(\frac{u}{2}\right)\left(1-p\right)\Big(1 +w \pi_{D}\left(p,1-p\right)\Big)\right] f_{N} ; \\
1-p &= \left[\left(\frac{u}{2}\right) p\Big(1 +w \pi_{C}\left(p,1-p\right)\Big) + \left(1-\frac{u}{2}\right)\left(1-p\right)\Big(1 +w \pi_{D}\left(p,1-p\right)\Big)\right] f_{N} ,
\end{align}
\end{subequations}
\end{linenomath}
which, in turn, holds if and only if the total population size, $N$ satisfies
\begin{linenomath}
\begin{align}\label{eq:baselineEquation}
f_{N} &= \frac{1}{1+w\left(R-1\right) p} .
\end{align}
\end{linenomath}
The right-hand-side of \eq{baselineEquation} is independent of $f_{N}$, and once this quantity is calculated, it is straightforward to check for any $f_{N}$ whether there exists $N$ for which \eq{baselineEquation} holds. If $f_{N}$ is strictly monotonic, then there exists at most one $N$ that satisfies this equation. For other types of baseline reproductive capacities, there might be several such $N$ that satisfy \eq{baselineEquation} (resulting in several non-zero metastable equilibria).

In addition to the simulations described in the main text, \textbf{Figs. \ref{fig:cooperatorsThenDefectors}--\ref{fig:defectorResistance}} demonstrate further effects of model parameters on metastable equilibria.

\begin{figure}
\centering
\includegraphics[width=0.8\textwidth]{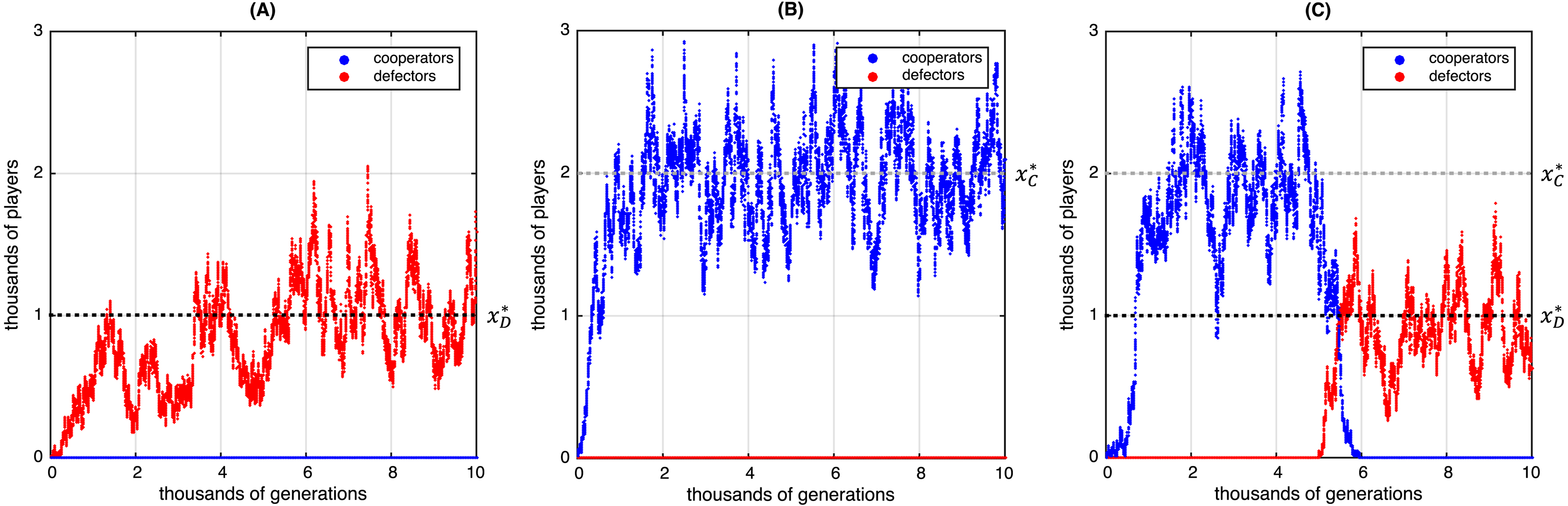}
\caption{Evolutionary game dynamics when there is no mutation and the baseline reproductive capacity, $f_{N}$, is defined by $f_{N} =\max\left\{0,c_{K}+r\left(1-\frac{N}{K}\right)\right\}$, where $c_{K}=1$, $r=0.005$, and $K=1000$. The multiplication factor for the public goods game is $R=1.5$ and the cost of cooperation is $w=0.01$. In (A), the population is initialized with no cooperators and 10 defectors. The defectors grow until they reach their carrying capacity of $1000$ and then persist at this metastable equilibrium ($f_{1}\geqslant 1$). In (B), the population is initialized with 10 cooperators and no defectors. The cooperators then grow in abundance until they reach their carrying capacity of approximately $2000$ players. It is immediate from panels (A) and (B) that groups of cooperators perform better than groups of defectors since selection allows them to reach a higher carrying capacity. In (C), the population is initialized with 10 cooperators and no defectors, and the population then proceeds to reach its carrying capacity. After $5000$ generations, an additional 10 defectors are introduced into the population, which disrupts the metastable equilibrium reached by the all-cooperator population. Defectors then outcompete and replace cooperators and finally reach their carrying capacity, which, as in (A), is significantly lower than the carrying capacity of an all-cooperator population.\label{fig:cooperatorsThenDefectors}}
\end{figure}

\begin{figure}
\centering
\includegraphics[width=0.8\textwidth]{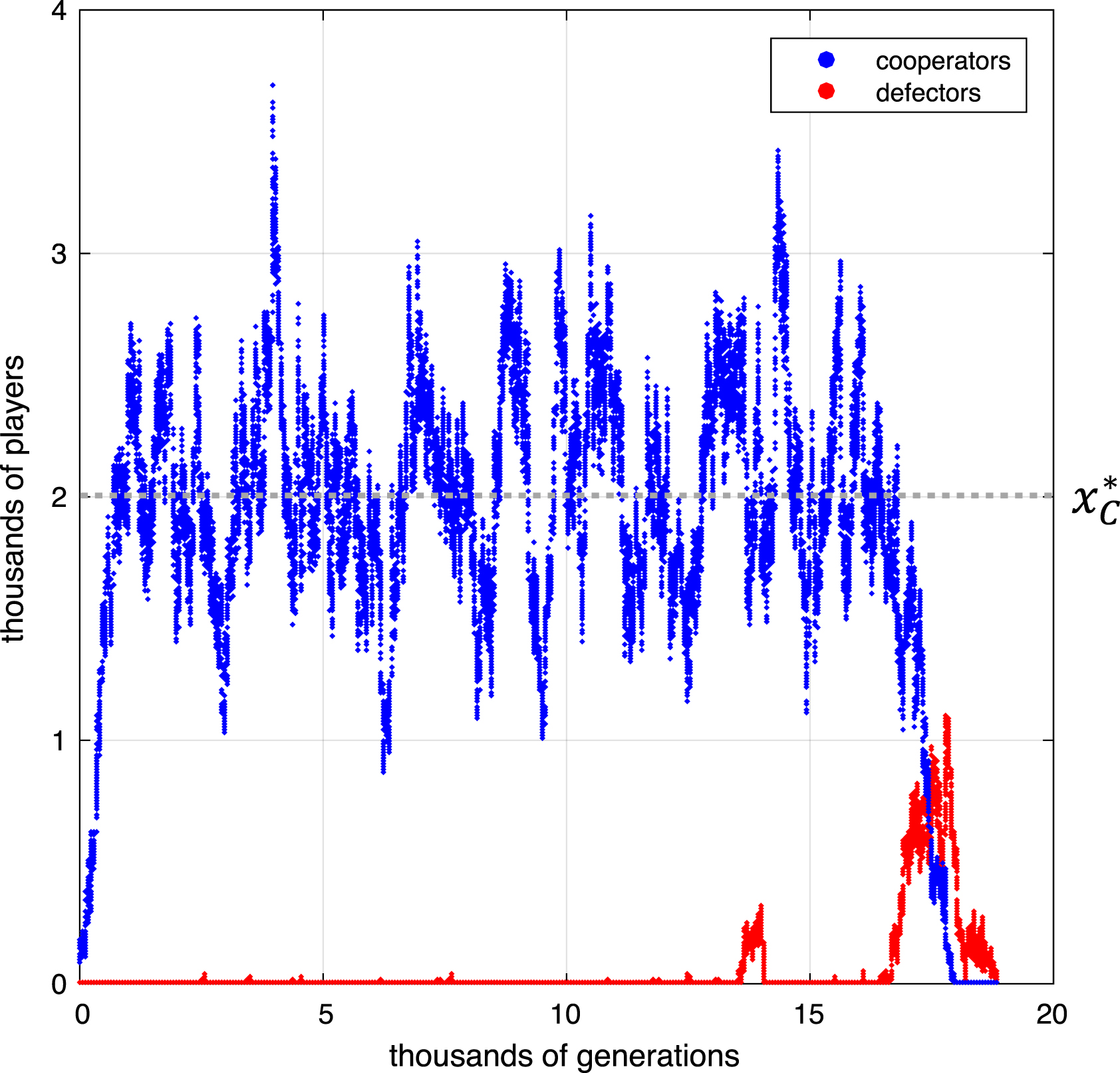}
\caption{Mutation-selection dynamics with drastically reduced mutation rates. The parameters used here are $u =0.00001$, $w=0.005$, and $R=2.0$. The baseline reproductive capacity is $f_{N} =\max\left\{0,c_{K}+r\left(1-\frac{N}{K}\right)\right\}$, where $c_{K}=0.995$, $r=0.005$, and $K=2000$. The population is initialized with $100$ cooperators and no defectors, and the cooperators then grow to reach their carrying capacity. Although small clusters of defectors are occasionally introduced through mutation, cooperators can resist invasion for a short period of time. Eventually, defectors invade and replace cooperators, and the population goes extinct due to the fact that defectors cannot survive on their own ($f_{1}<1$).\label{fig:defectorResistance}}
\end{figure}

\begin{lemma}
If $u\in\left(0,1\right)$ and $w\neq 0$, then $p$ is a strictly increasing function of $R$ with $p\uparrow 1/2$ as $R\rightarrow\infty$.
\end{lemma}
\begin{proof}
Since the polynomial defined by \eq{quadraticPolynomial} satisfies $\varphi\left(0\right) =u$ and $\varphi\left(1/2\right) =-w\left(1-u\right) /2$, we see that the solution to $\varphi\left(p\right) =0$ that falls within $\left[0,1\right]$ is actually at most $1/2$. Moreover, we can write
\begin{linenomath}
\begin{align}
\varphi_{R}\left(p\right) &:= \varphi\left(p\right) = 2\left(u R-1\right) w p\left(\frac{1}{2}-p\right) - \Big( \left(1-u\right) w+2u \Big) p + u ,
\end{align}
\end{linenomath}
where, notably, only the coefficient of $p\left(\frac{1}{2}-p\right)$ depends on $R$. Thus, if $R_{1}<R_{2}$ and $p_{1}$ satisfies $\varphi_{R_{1}}\left(p_{1}\right) =0$, then $\varphi_{R_{2}}\left(p_{1}\right) >0$. Since $\varphi\left(1/2\right) <0$, the unique solution to $\varphi_{R_{2}}\left(p_{2}\right) =0$ satisfies $p_{2}>p_{1}$, thus $p$ is an increasing function of $R$. That $p\uparrow 1/2$ as $R\rightarrow\infty$ follows immediately from taking the limit of \eq{equilibriumFraction}.
\end{proof}

\begin{theorem}\label{thm:criticalR}
Suppose that $f_{N}\downarrow 0$ as $N\rightarrow\infty$. If $u\neq 0$ and $w\neq 0$, then, for each $N\geqslant 1$, there is a critical multiplication factor, $R_{N}^{\ast}\geqslant 1$, which is the minimum multiplication factor for which there exists a non-zero metastable equilibrium supporting a population size of at least $N$ whenever $R\geqslant R_{N}^{\ast}$.
\end{theorem}
\begin{proof}
Since $p\uparrow 1/2$ as $R\rightarrow\infty$, we see that $\frac{1}{1+w\left(R-1\right) p}\downarrow 0$ as $R\rightarrow\infty$. Let
\begin{linenomath}
\begin{align}
R_{N}^{\ast} &:= \inf\left\{ R \geqslant 1 \mid f_{N'} = \frac{1}{1+w\left(R-1\right) p}\textrm{ for some }N'\geqslant N\right\} .
\end{align}
\end{linenomath}
Since $f_{N}\downarrow 0$ as $N\rightarrow\infty$, we have $R_{N}^{\ast}<\infty$. If $R\geqslant R_{N}^{\ast}$ and $N'$ satisfies $f_{N'} = \frac{1}{1+w\left(R_{N}^{\ast}-1\right) p}$, then $\left(x_{C}^{\ast},x_{D}^{\ast}\right) =\left(pN',\left(1-p\right) N'\right)$ is a metastable equilibrium by \eq{xCondition}. Furthermore, if $R\geqslant R_{N}^{\ast}$, then 
\begin{linenomath}
\begin{align}
\frac{1}{1+w\left(R-1\right) p} \leqslant \frac{1}{1+w\left(R_{N}^{\ast}-1\right) p} ,
\end{align}
\end{linenomath}
and it follows that any solution to $f_{N''} =\frac{1}{1+w\left(R-1\right) p}$ satisfies $N''\geqslant N'\geqslant N$, as desired.
\end{proof}

\begin{remark}
If $\inf_{N\geqslant 1}f_{N} >0$, then \thm{criticalR} need not hold. For example, whenever $R$ is sufficiently large and $w\neq 0$, we have $\frac{1}{1+w\left(R-1\right) p}<\inf_{N\geqslant 1}f_{N}$, so no value of $N$ satisfies \eq{baselineEquation}.
\end{remark}

\subsubsection{Variance}
In state $\left(x_{C},x_{D}\right)$, the expected squared number of cooperators in the next generation is
\begin{linenomath}
\begin{align}\label{eq:expectedSquared}
\mathbf{E}_{\left(x_{C},x_{D}\right)}\left[y_{C}^{2}\right] &= \sum_{\left(y_{C},y_{D}\right)}\mathbf{P}_{\textrm{WFB}}\left( y_{C},y_{D} \mid x_{C},x_{D} \right) m_{2} ,
\end{align}
\end{linenomath}
where
\begin{linenomath}
\begin{align}
m_{2} &= \sum_{a=0}^{y_{C}}\sum_{b=0}^{y_{D}}\left(a+b\right)^{2}\binom{y_{C}}{a}\left(1-\frac{u}{2}\right)^{a}\left(\frac{u}{2}\right)^{y_{C}-a}\binom{y_{D}}{b}\left(1-\frac{u}{2}\right)^{y_{D}-b}\left(\frac{u}{2}\right)^{b} \nonumber \\
&= \left(1-\frac{u}{2}\right)y_{C}\left( 1+\left(1-\frac{u}{2}\right)\left(y_{C}-1\right)\right) + 2\left(\frac{u}{2}\right)\left(1-\frac{u}{2}\right)y_{C}y_{D} + \left(\frac{u}{2}\right)y_{D}\left(1+\left(\frac{u}{2}\right)\left(y_{D}-1\right)\right) .
\end{align}
\end{linenomath}
It follows from a straightforward calculation that
\begin{linenomath}
\begin{align}
\mathbf{E}_{\left(x_{C},x_{D}\right)}\left[y_{C}^{2}\right] &= \left[\left(1-\frac{u}{2}\right)\left(x_{C}F_{C}\right) + \left(\frac{u}{2}\right)\left(x_{D}F_{D}\right)\right] + \left[\left(1-\frac{u}{2}\right)\left(x_{C}F_{C}\right) + \left(\frac{u}{2}\right)\left(x_{D}F_{D}\right)\right]^{2} \nonumber \\
&= \mathbf{E}_{\left(x_{C},x_{D}\right)}\left[y_{C}\right] + \mathbf{E}_{\left(x_{C},x_{D}\right)}\left[y_{C}\right]^{2} .
\end{align}
\end{linenomath}
Therefore, $\textbf{Var}_{\left(x_{C},x_{D}\right)}\left[y_{C}\right] =\mathbf{E}_{\left(x_{C},x_{D}\right)}\left[y_{C}\right]$, and, similarly, $\textbf{Var}_{\left(x_{C},x_{D}\right)}\left[y_{D}\right] =\mathbf{E}_{\left(x_{C},x_{D}\right)}\left[y_{D}\right]$. Thus,
\begin{linenomath}
\begin{subequations}
\begin{align}
\sigma_{\left(x_{C},x_{D}\right)}\left[y_{C}\right] / \mathbf{E}_{\left(x_{C},x_{D}\right)}\left[y_{C}\right] &= \frac{1}{\sqrt{\mathbf{E}_{\left(x_{C},x_{D}\right)}\left[y_{C}\right]}} ; \\
\sigma_{\left(x_{C},x_{D}\right)}\left[y_{D}\right] / \mathbf{E}_{\left(x_{C},x_{D}\right)}\left[y_{D}\right] &= \frac{1}{\sqrt{\mathbf{E}_{\left(x_{C},x_{D}\right)}\left[y_{D}\right]}} ,
\end{align}
\end{subequations}
\end{linenomath}
which both approach $0$ as $\mathbf{E}_{\left(x_{C},x_{D}\right)}\left[y_{C}\right]$ and $\mathbf{E}_{\left(x_{C},x_{D}\right)}\left[y_{D}\right]$ get large.

\section{Extinction time for branching games}\label{sec:extinctionTime}
We now characterize the extinction time for our model, inspired by techniques used in classical branching processes \citep[see][]{jagers:JAP:2011,faure:AAP:2014,schreiber:S:2017}. Let $x=\left(x_{C}/K,x_{D}/K\right)$ be normalized quantities of cooperators and defectors, where $K>0$ parametrizes the baseline reproductive capacity, $f_{N}$. Let
\begin{linenomath}
\begin{align}
A\left(x\right) := \begin{pmatrix}
\left(1-u\right) F_{C}\left(xK\right) & u F_{D}\left(xK\right) \\
u F_{C}\left(xK\right) & \left(1-u\right) F_{D}\left(xK\right)
\end{pmatrix} ,
\end{align}
\end{linenomath}
and consider the map, $\phi$, defined by
\begin{linenomath} 
\begin{align}
\phi &: \mathbb{R}_{\geqslant 0}^{2} \longrightarrow \mathbb{R}_{\geqslant 0}^{2} \nonumber \\
&: x \longmapsto A\left(x\right) x .
\end{align}
\end{linenomath}
Since $F_{C}$ and $F_{D}$ are bounded, so too is $\phi$.

We consider the normalized Markov chain $X_{t}^{K}=\left(C_{t}/K,D_{t}/K\right)$, where $C_{t}$ and $D_{t}$ are the number of cooperators and defectors at time $t$, respectively. Write $p_{K}\left(x,y\right) =\mathbf{P}\left[X_{t+1}^{K}=y \mid X_{t}^{K}=x\right]$ for the transition kernel. The transition probabilities are Poisson-distributed with mean given by the matrix $A\left(x\right)$. Note that $p_{K}\left(\mathbf{0},y\right) =0$ for all $y\neq\mathbf{0}$ since $\mathbf{0}$ is an absorbing state.

A measure $\mu_{K}\in\Delta\left(\mathbb{R}_{\geqslant 0}^{2}\--\mathbf{0}\right)$ is a quasi-stationary distribution for $p_{K}$ if there exists $\lambda_{K}\in\left(0,1\right)$ such that
\begin{linenomath}
\begin{align}
\int_{x\in\mathbb{R}_{\geqslant 0}^{2}\--\mathbf{0}} p_{K}\left(x,E\right) \, d\mu_{K}\left(x\right) = \lambda_{K}\mu_{K}\left(E\right) .
\end{align}
\end{linenomath}
for all $E\subseteq\mathbb{R}_{\geqslant 0}^{2}\--\mathbf{0}$. We denote the extinction time, i.e. the time until the chain is absorbed at the state $\mathbf{0}$, by $\tau_{K}$. Note that, if we start distributed according to a quasi-stationary distribution, $\mu_{K}$, then the probability of being absorbed in the next step is $1-\lambda_{K}$ since
\begin{linenomath}
\begin{align}
\mathbf{P}\left[ X_{t+1}^{K}=\mathbf{0} \mid  X_{t}^{K}\sim \mu_{K} \right] &= \int_{x\in\mathbb{R}_{\geqslant 0}^{2}\--\mathbf{0}} p_{K}\left(x,\mathbf{0}\right) \, d\mu_{K}\left(x\right) \nonumber \\
&= \int_{x\in\mathbb{R}_{\geqslant 0}^{2}\--\mathbf{0}} \Big( 1-p_{K}\left(x,\mathbb{R}_{\geqslant 0}^{2}\--\mathbf{0}\right) \Big) \, d\mu_{K}\left(x\right) \nonumber \\
&= 1-\lambda_{K} .
\end{align}
\end{linenomath}
Moreover, if not absorbed in the next time step, the chain remains distributed according to $\mu_{K}$. Therefore, the extinction time $\tau_{K}$ is a geometric random variable with parameter $1-\lambda_{K}$, and $\mathbf{E}\left[\tau_{K}\right] =1/\left(1-\lambda_{K}\right)$, where $\mathbf{E}\left[\tau_{K}\right]$ denotes the expected value of $\tau_{K}$ when the chain is initially distributed according to $\mu_{K}$.

\begin{proposition}
There exists $c>0$, independent of $K$, such that $\mathbf{E}\left[\tau_{K}\right]\geqslant e^{cK}$.
\end{proposition}
\begin{proof}
For the model considered in the main text, $\phi$ has a unique fixed point, $x^{\ast}$, with $\phi\left(x^{\ast}\right) =x^{\ast}$. Moreover, this fixed point is an attractor. (One can show that the normalized quasi-stationary distribution $\mu_{K}$ converges weakly to $\delta_{x^{\ast}}$ as $K\rightarrow\infty$.) Therefore, there exists $\delta >0$ and an open set, $U$, containing $x^{\ast}$ such that $N_{\delta}\left(\phi\left(\overline{U}\right)\right)\subseteq U$, where for $E\subseteq\mathbb{R}_{\geqslant 0}^{2}-\mathbf{0}$, the $\delta$-neighborhood of $E$ is defined as
\begin{linenomath}
\begin{align}
N_{\delta}\left(E\right) &:= \left\{y\in\mathbb{R}_{\geqslant 0}^{2}-\mathbf{0} \mid \inf_{x\in E}\left\| y-x\right\| <\delta\right\} .
\end{align}
\end{linenomath}
By definition of the quasi-stationary distribution, $\mu_{K}$, we have
\begin{linenomath}
\begin{align}
\lambda_K \mu_{K}\left(U\right) &= \int_{x\in\mathbb{R}_{\geqslant 0}^{2}\--\mathbf{0}} p_K(x,U) \;d\mu_K(x) \nonumber \\
&\geqslant \int_{x\in U} p_{K}\left(x,U\right) \, d\mu_{K}\left(x\right) \nonumber \\
&\geqslant \mu_{K}\left(U\right) \inf_{x\in\overline{U}} p_{K}\left(x,U\right) \nonumber \\
&= \mu_{K}\left(U\right)\left(1-\sup_{x\in\overline{U}} p_{K}\left(x,U^{c}\right) \right) . 
\end{align}
\end{linenomath}
For $x\in\overline{U}$, $\phi\left(x\right)\in \phi\left(\overline{U}\right)$, which implies that $N_{\delta}\left(\phi\left(x\right)\right)\subset N_{\delta}\left(\phi\left(\overline{U}\right)\right)\subset U$. Therefore,
\begin{linenomath}
\begin{align}
\lambda_{K} &\geqslant 1-\sup_{x\in\overline{U}} p_{K}\left(x,U^{c}\right) \nonumber \\
&\geqslant 1-\sup_{x\in\overline{U}} p_{K}\left(x,N_{\delta}\left(\phi\left(x\right)\right)^{c}\right) .
\end{align}
\end{linenomath}
To complete the proof, we bound $p_{K}\left(x,N_{\delta}\left(\phi\left(x\right)\right)^{c}\right)$ via a large-deviation estimate based on the Chernoff-Cramer method. If $Z$ is a Poisson random variable with mean $\phi\left(x\right) K$, then
\begin{linenomath}
\begin{align}
p_{K}\left(x,N_{\delta}\left(\phi\left(x\right)\right)^{c}\right) &= \mathbf{P}\left[X_{t+1}^{K}\not\in N_{\delta}\left(\phi\left(x\right)\right) \mid X_{t}^{K} = x\right] \nonumber \\
&= \mathbf{P}\left[ \left| X_{t+1}^{K}-\phi\left(x\right)\right| > \delta \mid X_{t}^{K} = x\right] \nonumber \\
&= \mathbf{P}\left[ \left| Z-\phi\left(x\right) K\right| > \delta K \right] .
\end{align}
\end{linenomath}
Using Markov's inequality and the Poisson moment-generating function, we see that
\begin{linenomath}
\begin{align}
\mathbf{P}\left[ Z > \left(\phi\left(x\right) +\delta\right) K\right] \leqslant \frac{\mathbf{E}\left[e^{tZ}\right]}{e^{t\left(\phi\left(x\right) +\delta\right)K}} = \frac{e^{\phi\left(x\right) K\left(e^{t}-1\right)}}{e^{t\left(\phi\left(x\right) +\delta\right) K}}.
\end{align}
\end{linenomath}
As a function of $t$, the minimum of $\frac{e^{\phi\left(x\right) K\left(e^{t}-1\right)}}{e^{t\left(\phi\left(x\right) +\delta\right) K}}$ is at $t^{\ast}=\log\left(1+\delta /\phi\left(x\right)\right)$. Since the function $g\left(y\right) :=\log\left(1+\delta /y\right)\left(y+\delta\right) -\delta$ satisfies $g\left(y\right) >0$ and $g'\left(y\right) <0$ for all $y>0$, we have
\begin{linenomath}
\begin{align}
\mathbf{P}\left[ Z > \left(\phi\left(x\right) +\delta\right) K\right] \leqslant e^{-g\left(\phi\left(x\right)\right) K} \leqslant e^{-g\left(m\right) K} ,
\end{align}
\end{linenomath}
where $m=\max_{x\in\mathbb{R}_{\geqslant 0}^{2}}\phi\left(x\right)$. It follows that with $c:=g\left(m\right)$,
\begin{linenomath}
\begin{align}
\mathbf{E}\left[\tau_{K}\right] &= \frac{1}{1-\lambda_{K}} \geqslant e^{cK} ,
\end{align}
\end{linenomath}
which completes the proof.
\end{proof}

\section*{Acknowledgments}
This work was supported by the Office of Naval Research, grant N00014-16-1-2914. C. H. acknowledges financial support from the Natural Sciences and Engineering Research Council of Canada (NSERC), grant RGPIN-2015-05795. The Program for Evolutionary Dynamics is supported, in part, by a gift from B. Wu and Eric Larson.

\end{document}